\documentclass[aip,jcp,reprint]{revtex4-1}
\usepackage{amsmath,amssymb,amsthm,dsfont,color,graphicx,listings,dcolumn,epstopdf}
\usepackage[mathscr]{euscript}
\usepackage[dvipsnames]{xcolor}
\usepackage{soul}

\newcommand{\<}{\left \langle}
\renewcommand{\>}{\right \rangle}

\renewcommand{\P}{\mathrm{\mathbf{P}}}
\newcommand{\1}{\mathds{1}}
\newcommand{\N}{\mathbf{N}}
\def\avevec{v}
\DeclareMathOperator{\var}{var}
\DeclareMathOperator{\cov}{cov}

\newcommand{\convdist}{\overset{\rm d}{\to}}
\newcommand{\erf}{\mathrm{erf}}

\newcommand{\Real}{\mathbb{R}}
\newcommand{\eps}{\varepsilon}
\newtheorem{theorem}{Theorem}[section]

\newtheorem{assumption}[theorem]{Assumption}
\newtheorem{lemma}[theorem]{Lemma}

\renewcommand{\t}{{\textrm{\tiny \rm T}}}

\newcommand{\us}{{\text{\tiny EMUS}}}
\newcommand{\zmb}{z^{\text{\tiny  MBAR}}}
\newcommand{\mb}{{\text{\tiny  MBAR}}}

\newcommand{\ignore}[1]{}

\begin{document}


\title[Eigenvector Method for Umbrella Sampling]{Eigenvector method for umbrella sampling enables error analysis}

\author{Erik Thiede}
\thanks{Equal contributions}
\affiliation{Department of Chemistry and James Franck Institute, the University of Chicago, Chicago, IL 60637}
\author{Brian Van Koten}
\thanks{Equal contributions}
\affiliation{Department of Statistics and James Franck Institute, the University of Chicago, Chicago, IL 60637}
\author{Jonathan Weare}
\email{weare@uchicago.edu}
\affiliation{Department of Statistics and James Franck Institute, the University of Chicago, Chicago, IL 60637}
\author{Aaron R.\ Dinner}
\email{dinner@uchicago.edu}
\affiliation{Department of Chemistry and James Franck Institute, the University of Chicago, Chicago, IL 60637}

\begin{abstract}
Umbrella sampling efficiently yields equilibrium averages that depend on exploring rare states of a model by biasing simulations to windows of coordinate values and then combining the resulting data with physical weighting.  Here, we introduce a mathematical framework that casts the step of combining the data as an eigenproblem.  The advantage to this approach is that it facilitates error analysis. We discuss how the error scales with the number of windows.  Then, we derive a central limit theorem for averages that are obtained from umbrella sampling.  The central limit theorem suggests an estimator of the error contributions from individual windows, and we develop a simple and computationally inexpensive procedure for implementing it.  We demonstrate this estimator for simulations of the alanine dipeptide and show that it emphasizes low free energy pathways between stable states in comparison to existing approaches for assessing error contributions.  We discuss the possibility of using the estimator and, more generally, the eigenvector method for umbrella sampling to guide adaptation of the simulation parameters to accelerate convergence.
\end{abstract}





\maketitle
\footnotetext{Second footnote}
\section{Introduction}
One of the main uses of molecular simulations is the calculation of equilibrium averages.  For understanding reaction processes, the free energy projected onto selected coordinates (collective variables) is of special interest.  It relates directly to the probabilities of the coordinates taking particular values, and  it can provide valuable information about the stable states, the barriers between them, and the origin of their stabilization.  Furthermore, it is the starting point for most rate theories.  Although in principle the free energy can be estimated from a long unbiased simulation, in practice doing so is challenging because bottlenecks slow the exploration of the configuration space.  In other words, transitions between regions of the space are very infrequent in comparison to local fluctuations.

Various methods have been introduced to overcome this problem.  Here, we consider one of the oldest and still most widely used such methods, umbrella sampling (US).\cite{torrie1977nonphysical}  In this approach, the collective-variable interval of interest is covered by a series of simulations, in each of which the system is biased such that sampling is restricted to a relatively narrow window of values of the collective variables. This can be accomplished by addition of a biasing potential that is small in the window and large outside it.  The information from the different simulations must be combined, and the effect of the bias removed, to obtain the overall free energy profile.  This requires consistently normalizing the probabilities in different windows, a  task  that is complicated by the fact that the simulations are run independently.


Considerable effort has been devoted to determining how best to combine the results from different simulations.  Initially, researchers manually adjusted the zero of free energy in each window to make the full free energy profile continuous and, often, smooth; conflicting results arising from limited sampling at the window peripheries were removed.  The desire to use all the simulation data motivated the introduction of estimators that allow for systematically combining the data from different simulations.  By far, the most widely used of these in chemical physics applications is the weighted histogram analysis method (WHAM). The multistate Bennett acceptance ratio (MBAR)  method, as it is referred to in the the molecular-simulation literature and will be referred to here, is closely related but does not rely on binning the data.\cite{vardi1985empirical,gill1988,kumar1992weighted,shirts2008statistically}  Both WHAM and MBAR can be derived from maximum-likelihood or minimum asymptotic variance principles assuming independent, identically distributed sampling in each window, and have  corresponding statistical optimality properties under those conditions.  Recent extensions seek to improve performance when the sampling is limited and to extend the algorithm to more general ensembles.\cite{rosta2014free,mey2014xtram}


In the present paper, we introduce an alternative scheme for estimating the free energy from US simulation data.  In this approach, the normalization constants needed to combine information from separate simulation windows are the components of the eigenvector of a stochastic matrix that can be constructed from running averages in the windows.  We thus term our method Eigenvector Method for Umbrella Sampling (EMUS).  The advantage of our method is that it lends itself to error analysis.  Following previous work,\cite{tan2012theory,lelievre2010free,minh2009optimal} we measure error with the asymptotic variance.

Our paper is organized as follows.  After giving some background on US in Section \ref{sec:background}, we formulate EMUS in Section \ref{sec:EMUS_basics}.  In Section \ref{ssec:numeric_comparison}, we show that EMUS performs comparably to WHAM and MBAR, and discuss its connection with the latter.  In Section \ref{sec:scaling}, we use  scaling arguments with simplifying assumptions to show that accounting for the error associated with combining the data is important and limits the speedups that can be achieved by increasing the number of simulation windows.  In Section \ref{sec:analysis}, we provide the full numerical analysis, which applies generally, without simplifying assumptions.  Specifically, we derive a central limit theorem for averages from EMUS and use it to develop a means for estimating the error contributions from individual windows.  We demonstrate the method for the free energy projected onto the $\phi$ and $\psi$ dihedral angles of the alanine dipeptide and compare the error contributions with those from an estimator introduced by Zhu and Hummer \cite{zhu2012convergence}.  We conclude in Section \ref{sec:conclusions}.

\section{Background on Umbrella Sampling} \label{sec:background}

Here, we review umbrella sampling and establish basic terms and notation.
The goal
is the calculation of an average of an observable $g$ over a time-independent
probability distribution $\pi$:
\begin{equation}\label{eq:thermo_avg}
\<g\> = \int g(x) \pi(x) dx.
\end{equation}
At thermal equilibrium, $\pi$ is the Boltzmann distribution:
\begin{equation}
\pi(x) = \frac{\exp \left(- H_0(x) / k_B T \right)}{\int \exp \left(- H_0(x)/ k_B T \right) dx},
\end{equation}
where $H_0$ is the system Hamiltonian,  $k_B$ is Boltzmann's constant, and $T$ is the temperature.
In particular, we can express the the free energy difference between two states $S_1$ and $S_2$ as
\begin{equation}\label{eq:defn_FE}
\Delta G = -k_B T \ln \left( \frac{\< \1_{S1} \>}{ \< \1_{S2} \>}\right),
\end{equation}
where  $\1$ is the indicator function
\begin{equation}\label{eq:defn_indicator}
\1_S(x) =
	\begin{cases}
		1 & \text{if }x \in S, \text{ and } \\
		0 & \text{otherwise}.
	\end{cases}
\end{equation}
Similarly, the reversible work to constrain a collective variable $q(x)$ to a particular value $q'$, also known as the potential of mean force (PMF), may be written as
\begin{equation}\label{eq:defn_FE_surface}
W(q') = - k_B T \ln \<\delta\left(q-q'\right)\>.
\end{equation}

For complex systems,
averages of the form in \eqref{eq:thermo_avg} must be evaluated numerically.  Typically, this is done by generating a chain of related configurations, $X_t$, using Monte Carlo methods or molecular dynamics, and by assuming ergodicity.  Namely, as the number of configurations $N$ goes to infinity,  $\< g\>$ is the limit of the sample mean:
\begin{equation}\label{eq:defn_sample_mean}
  \bar{g} = \frac{1}{N} \sum_{t=0}^{N-1}g(X_t).
\end{equation}
In all practical sampling methods, successive configurations are strongly correlated.
While ergodicity guarantees that sample means converge to averages over $\pi$,
convergence can be extremely slow if the correlation between subsequent points is strong.
This is the case when sampling $\pi$ relies on visiting low-probability states, such as transition states of chemical reactions.

US methods address this issue by enforcing sampling of different regions of configuration space (windows), introducing $L$ nonnegative \textit{bias functions} $\psi_i$ and then using $L$ independent simulations to sample from the biased probability distributions
\begin{equation}\label{eq:defn_pi_i}
\pi_i(x) \propto \psi_i(x) \pi(x).
\end{equation}
The essential idea is that sampling each $\pi_i$ is fast because $\psi_i$ is chosen so that relatively likely states under $\psi_i$ are not separated by relatively unlikely states.  This is accomplished by restricting the set of states on which $\psi_i$ is non-negligible so that $\pi$ is closer to constant on that set.  In Section~\ref{sec:low_T_limit} we make this point more carefully by examining   a regime in which umbrella sampling can be shown to be exponentially more efficient than direct simulation.
A popular choice is to use bias functions that take a Gaussian form:
\begin{equation}\label{eq:harmonic_bias_fxn}
\psi_i(q) = \exp \left( -\frac{1}{2}k_i \left(q-q^0_i\right)^2/k_B T\right),
\end{equation}
such that
\begin{equation}\label{eq:harmbf_biased_density}
\pi_i (x) \propto \exp \left[- \left( H_0(x)+\frac{1}{2} k_i \left(q-q^0_i\right)^2\right) / k_B T \right].
\end{equation}
This corresponds to adding a harmonic potential centered at $q_i^0$ with spring constant $k_i$ to the system Hamiltonian.
We call the relative normalization constant (or partition function) of the $i$-th biased distribution $z_i$:
\begin{equation}\label{eq:zi}
 z_i = \frac{\int  \psi_i(x) \pi(x) dx}{\sum_{k=1}^L \int \psi_k(x) \pi(x) dx} .
\end{equation}
We also define the free energy in window $i$ as
\begin{equation} \label{eqn:windowfreeenergy}
G_i = - k_B T \ln z_i.
\end{equation}
We denote averages over the biased distributions by
\begin{equation*}
  \< g\>_i = \int g(x) \pi_i(x) \, dx.
\end{equation*}
Overall averages of interest, $\langle g\rangle$, can be estimated as $z_i$-weighted sum of averages computed in each of the windows.  We detail our prescription in the next section.

\section{The Eigenvector Method for Umbrella Sampling}\label{sec:EMUS_basics}
In this section, we present the Eigenvector Method for Umbrella Sampling (EMUS).
We begin by defining
\begin{equation*}
g^\ast \equiv \frac{g}{\sum_{k=1}^L \psi_k}.
\end{equation*}
for any function $g$. Then, we observe that
\begin{align}
  \<g\> &= \int g(x) \pi(x) \, dx \nonumber \\
         &= \int g(x)\left\{ \frac{\sum_{i=1}^L \psi_i(x)  \left[\frac{\int\psi_i(x)\pi(x)dx}{\int\psi_i(x)\pi(x)dx}\right]}{\sum_{k=1}^L \psi_k(x)}\right\} \pi(x) \, dx
         \nonumber \\
         \nonumber \\
         &= \sum_{i=1}^L \int \psi_i(x) \pi(x) dx\frac{\int g^\ast(x) \psi_i(x) \pi(x)dx}{\int \psi_i(x) \pi(x) dx}
         \nonumber \\
         &= \sum_{i=1}^L z_i\left(\sum_{k=1}^L \int \psi_k(x) \pi(x) dx\right) \< g^\ast\>_i \label{eq:combined_avg0}.
\end{align}
The factor in parentheses can be taken out of the sum over $i$.  To express this factor in terms of computable averages, we repeat the same steps with $g=1$:
\begin{equation}\label{eqn:1star}
\sum_{k=1}^L \int \psi_k(x) \pi(x) dx = \frac{1}{\sum_{i=1}^L z_i \< 1^\ast\>_i}.
\end{equation}
Substituting \eqref{eqn:1star} into \eqref{eq:combined_avg0},
\begin{equation}\label{eq:combined_avg}
\<g\> = \frac{\sum_{k=1}^L z_i \<g^*\>_i}{\sum_{k=1}^L z_i \<1^*\>_i}.
\end{equation}
Consequently, if we can evaluate the $z_i$ and the $\<g^*\>_i$ then we can assemble the original average $\<g\>$ of interest.  The averages  $\<g^*\>_i$ can be computed by sequences $X^i_t$ (typically independent for each $i$) that sample $\pi_i$.
Umbrella sampling methods differ primarily in how the $z_i$ are computed.



To express the constants $z_i$ in terms of averages over the biased distributions, we
take $g(x) = \psi_j(x)$ in~\eqref{eq:combined_avg0}.  Then, $z_i$ solves
\begin{equation}\label{eq:eval_for_z}
z_j = \sum_{i=1}^L z_i F_{ij},
\text{ where } F_{ij} = \< \psi_j^\ast \>_i.
\end{equation}
That is, the vector of normalization constants $z$ is a left eigenvector of the matrix $F$ with eigenvalue one.
Under conditions to be elaborated upon in Section \ref{sec:eigenproblem}, the solution to \eqref{eq:eval_for_z} is uniquely specified when we notice that
\begin{equation}
\sum_{i=1}^L z_i =1.
\label{eq:normalization of z}
\end{equation}
\subsection{Computational Procedure}
In the EMUS algorithm, we estimate the entries of $F$ and the
averages $\< g^\ast \>_i$ and  $\< 1^\ast \>_i$  by sample means, then assemble the estimate
of $\<g\>$ using \eqref{eq:combined_avg}.
To be precise, we denote the sample means by
\begin{align*}
\bar{g}^*_i &= \frac{1}{N_i}\sum_{t=0}^{N_i-1} \frac{g(X_t^i)}{\sum_k \psi_k (X_t^i)}, \\
\bar{1}^*_i &= \frac{1}{N_i}\sum_{t=0}^{N_i-1} \frac{1}{\sum_k \psi_k (X_t^i)},
\text{ and } \\
\bar{F}_{ij} &= \frac{1}{N_i}\sum_{t=0}^{N_i-1} \frac{\psi_j(X_t^i)}{\sum_k \psi_k (X_t^i)}.
\end{align*}
EMUS proceeds as follows:
\begin{enumerate}
\item Choose the biasing functions $\psi_i$.
\item Compute trajectories that sample states $X_t^i$ from the biased distributions $\pi_i$.
\item Calculate the matrix $\bar{F}$ and the averages $\bar{g}_i^*$ and $1_i^\ast$.
\item Calculate the vector of estimated normalization constants $z^\us$ as the solution to 
\[
z_j^\us = \sum_{i=1}^L z_i^\us \bar F_{ij}\quad\text{with}\quad \sum_{i=1}^L z_i^\us =1.
\]
      We use QR factorization as given by Golub and Meyer.\cite{golub1986using}  See also Section \ref{secn:algorithm}.
\item Compute the estimate of $\<g\>$:
\begin{equation}\label{eq:emus estimator}
\langle g \rangle^{\us} =
\frac{\sum_{i=1}^L z_i^\us\, \bar g^\ast_i}{\sum_{i=1}^L z_i^\us\,  \bar 1^\ast_i}
\end{equation}
by substituting $z^\us$ and the sample means in~\eqref{eq:combined_avg}.
\end{enumerate}

We remark that when one wishes to compute a free energy difference or a ratio of two observables, as in equation~\eqref{eq:defn_FE},
it is not necessary to compute $1^\ast_i$.
Instead, one may use the formula
\begin{equation*}
  \frac{\langle \mathds{1}_{S1}\rangle }{\langle \mathds{1}_{S2} \rangle} =
  \frac{\sum_{i=1}^L z_i\, \langle \mathds{1}_{S1} \rangle_i}{\sum_{i=1}^L z_i\, \langle \mathds{1}_{S2}\rangle_i}
\end{equation*}
in place of~\eqref{eq:emus estimator}.

\subsection{The Eigenvector Problem}\label{sec:eigenproblem}

In this section, we give conditions under which the eigenvector problem has a unique solution.
First, we show that $F$ is a stochastic matrix;
that is, each element $F_{ij}$ is nonnegative and every row of $F$ sums to one.  For the latter,
\begin{equation*}
\sum_{j=1}^L F_{ij} = \sum_{j=1}^L   \< \frac{\psi_j}{\sum_{k=1}^L \psi_k}\>_i
=  \<\frac{\sum_{j=1}^L  \psi_j}{\sum_{k=1}^L \psi_k}\>_i
=1.
\end{equation*}
The entries of $F$ are nonnegative since we require that the bias functions be nonnegative.
One can show that the matrix $\bar F$ is also stochastic by similar arguments.

A stochastic matrix $J$ has a unique eigenvector with eigenvalue one
if and only if it is irreducible: for every possible grouping of the indices into two distinct sets, $A$ and $B$,  $J_{ij} \neq 0$ for some $i\in A$ and $j\in B$.\cite{schneider1977concepts}
In fact, this statement remains true when $J$ is   non-negative with largest eigenvalue equal to one.  For any such matrix we let $z(J)$ denote the continuous function returning the unique left eigenvector of $J$ corresponding to eigenvalue one.

In the case of the particular stochastic matrix $F$ defined in \eqref{eq:eval_for_z} these statements imply that if, 
 for any division of the indices into sets $A$ and $B,$ there is sufficient overlap between the sets $\cup_{i\in A} \{x:\, \psi_i(x)>0\}$ and
$\cup_{j\in B} \{x:\, \psi_j(x)>0\}$ then there will be a unique solution $z(F)$ to \eqref{eq:eval_for_z} which will necessarily equal the relative normalization constants $z$ defined in \eqref{eq:zi}.
Because  $z(J)$ is a continuous function of its arguments, $z^\us= z(\bar F)$ converges to $z$ as $\bar F$ converges to $F.$  Consequently, EMUS produces a consistent estimator in the sense that if the sample averages used to estimate the entries  $F_{ij}$ and $\<g^*\>_i$ converge (in probability or with probability one) to the true values, then the estimate of $\<g\>$ also converges (in the same sense). 

\section{The Connection between EMUS and MBAR}

Building upon earlier work in the statistics literature\cite{vardi1985empirical, gill1988,tan2004likelihood}, Shirts and Chodera\cite{shirts2008statistically} suggested a class of algorithms for estimating free energy differences between states, which they termed MBAR.  This method is similar to WHAM but does not require binning the simulation data to form histograms (see Tan {\it et al.}\cite{tan2012theory}).
In this section, we explain the relation between EMUS and MBAR \cite{shirts2008statistically}.  We also derive a new iterative method for solving the MBAR equations,
and we show that our iteration leads naturally to a new family of related consistent estimators.

Shirts and Chodera's\cite{shirts2008statistically} starting point is the identity (see their (5))
\begin{equation}\label{eq:MBAR_family}
z_j \sum_{i=1}^L \<\alpha_{ij}(x)\psi_i(x)\pi(x)\>_j= \sum_{i=1}^L z_i \<\alpha_{ij}(x)\psi_j(x)\pi(x)\>_i,
\end{equation}
where $\alpha_{ij} (x)$ is an arbitrary function.  
They proposed the choice
\begin{equation}\label{eq:MBAR_estimator}
\alpha_{ij}^\mb (x) = \frac{n_i/ z_i}{\sum_k  \psi_k (x) \pi(x)n_k /z_k},
\end{equation}
where $n_i$ is the number of uncorrelated samples in window $i$.
Substituting \eqref{eq:MBAR_estimator} into \eqref{eq:MBAR_family} gives
\begin{equation}\label{eq:MBAR_eval}
z_j = \sum_{i=1}^L z_i  \< \frac{\psi_j (x) n_i / z_i }{\sum_k \psi_k (x) n_k / z_k } \>_i.
\end{equation}

We can cast \eqref{eq:MBAR_eval} in a form reminiscent of EMUS by writing
\begin{equation}\label{eq:MBAR_eval2}
z_j = \sum_{i=1}^L z_i F_{ij}(z),
\end{equation}
where
\begin{equation}\label{eq:MBAR_F}
F_{ij}(w) =  \< \frac{\psi_j (x)n_i  / w_i }{\sum_k \psi_k (x) n_k / w_k } \>_i
\end{equation}
for any vector $w$ with positive entries.
EMUS corresponds to setting $w = n$ so that
\begin{equation*}
\alpha_{ij}^\us (x) = \frac{1}{  \sum_k \pi(x) \psi_k (x)},
\end{equation*}
and \eqref{eq:MBAR_family} reduces to the eigenproblem~\eqref{eq:eval_for_z}.

In practice, one must replace the matrix $F_{ij}(w)$ in \eqref{eq:MBAR_F} by the sample mean approximation
\begin{equation}\label{eq:MBAR_Fapprox}
\bar{F}_{ij}(w) =   \frac{1}{N_i}\sum_{t=0}^{N_i-1} \left[\frac{\psi_j \left( X_t^i \right)n_i  / w_i }{\sum_k  \psi_k \left( X_t^i \right)n_k  /w_k}\right].
\end{equation}
Substituting $\bar F_{ij}(z)$ for $F_{ij}(z)$ in \eqref{eq:MBAR_eval2} yields the equation
\begin{equation}\label{eq:MBAR_approx}
\zmb_j = \sum_{i=1}^L \zmb_i \bar F_{ij}({\zmb})
\end{equation}
for $\zmb$, which we refer to here as the MBAR estimator.
If the samples $X_t^i$ are independent, MBAR is the nonparametric maximum-likelihood estimator of $z$.\cite{vardi1985empirical}

In practice,  the samples $X^i_t$ are not independent for a given $i$, and 
the $n_i$  must be estimated from data.  Several algorithms for estimating the $n_i$ have been proposed.\cite{frenkel2001understanding,geyer1992practical} Shirts and Chodera\cite{shirts2008statistically} base their estimates on the integrated autocorrelation times of physically-motivated coordinates, and we follow this common practice here.
In fact,
once the $n_i$ have been estimated,
Shirts and Chodera\cite{shirts2008statistically}   suggest replacing sample averages over all $N_i$ points by sample averages over the $n_i$ points obtained by including only every $N_i/n_i$-th sample along the trajectory.
We note that both the subsampling approach and the one in \eqref{eq:MBAR_Fapprox} correspond to approximations of expression \eqref{eq:MBAR_family} with \eqref{eq:MBAR_estimator}, and we regard both as variations on the MBAR estimator.
When the samples are independent, the two approaches are the same.
In tests of the iterative EMUS algorithm introduced below, we find estimates to be insensitive to the choice of $n_i$ and they can be set equal to 1, though in that case the estimator no longer corresponds directly to MBAR.


As written above, the MBAR estimator \eqref{eq:MBAR_approx} resembles an eigenvector problem.  However, the dependence of $\bar F(z)$ on $z$ implies that  the solution must be obtained self-consistently.
{The approach advocated by Shirts and Chodera for computing the MBAR estimator corresponds in the framework described here to solving \eqref{eq:MBAR_approx} by a Newton-type iteration.}
However, the eigenvector form of \eqref{eq:MBAR_approx} suggests an alternative approach.
Rather than Newton's method,
we employ the following algorithm:
\begin{enumerate}
\item  As an initial guess for $\zmb$, choose a vector $z^{0}$ with positive entries.
 Estimate the $n_i$. Set $m=0$.
\item{
\begin{enumerate}
\item Calculate $\bar F_{ij}(z^{m})$ according to \eqref{eq:MBAR_Fapprox}.
\item Calculate a new estimate $z^{m+1}$ of $\zmb$ by solving the eigenproblem
\begin{equation}\label{eq:EMUS_it}
z_j^{m+1} = \sum_{i=1}^L z_i^{m+1} \bar F_{ij}(z^{m}).
\end{equation}
\end{enumerate}
}
\item{If  $\max_i |z_i^{m+1}-z_i^{m}|/z_i^{m} > {\rm Tolerance}$,
\begin{enumerate}\item Increment $m$; \item Go to Step 2. \end{enumerate}
}
\end{enumerate}

To show that this iteration makes sense,
we must prove that the eigenproblem~\eqref{eq:EMUS_it}
always has a unique solution and that $z^m$ converges to $\zmb$ as $m$ goes to infinity.
To see that the eigenproblem has a solution,
first observe that if $\bar{F}_{ij}(w)$ is irreducible for one vector with positive entries,
$w$, then it is irreducible for all vectors with positive entries.
When applying the EMUS method,
we assume that  both $F= F(n)$ and $\bar{F} = \bar{F}(n)$ are irreducible.
Thus, we may assume that $\bar{F}(w)$ is irreducible for any $w$.
Moreover, observe that for any positive vector $w$, the vector with entries $n_i/w_i$  
is a right eigenvector of $\bar{F}(w)$ with eigenvalue one and positive entries.
It follows from the Perron-Frobenius theorem that the matrix $\bar{F}(w)$
has a unique left eigenvector $z(\bar{F}(w))$ with eigenvalue one and that $z(\bar{F}(w))$ has positive entries.
Thus, the eigenproblem always has a unique solution.
We do not have a proof that the iterates converge.
However, since $\zmb$ is a fixed point of the iteration,
if the iterates do converge their limit must be $\zmb$.
In practice, we find that the iteration converges quickly,
usually to a relative error of $10^{-6}$ within 10 iterates.

In addition to its apparently rapid convergence, another argument in favor of the algorithm that we introduce above for solving \eqref{eq:MBAR_approx}
is that each iteration of the scheme results in a new consistent estimator.
We will use the term iterative EMUS to refer to this family of estimators.
With the initial guess $z^0 = n$,  
the result, $z^1$, of the first iteration is the EMUS estimator defined in Section \ref{sec:EMUS_basics}.
In Appendix~\ref{apdx: consistency}, we show that for any fixed finite number of iterations $m$, $z^m$ is also a consistent
estimator of the vector $z$ of normalization constants.  
By contrast, other schemes, such as Newton's method, for solving~\eqref{eq:MBAR_approx}
may require that the number of iterations goes to infinity to obtain a consistent estimate.
{We also remark that the consistency result in Appendix~\ref{apdx: consistency} holds as long as the $n_i$ converge  to non-random, positive values with increasing numbers of samples $N_i$.  They can be chosen as described above, or simply set to a  fixed value.}

Differences between the iterative EMUS scheme above
and the application of Newton's method proposed by Shirts and Chodera~\cite{shirts2008statistically}
are mostly matters of implementation.
As we will see in the next section, the results are not very sensitive to these computational details;
most of the accuracy in the iterative EMUS approach is achieved in the first step.
In any case, we remind the reader that the primary goal of this paper is to characterize those properties of the broader umbrella sampling approach that are essential to its success, not to analyze details of implementation.

While we focus here on potentials of mean force,
the MBAR estimator has been applied to a broader category of free energy problems,
including the analysis of single-molecule pulling experiments and alchemical free energy calculations.\cite{shirts2008statistically,paliwal2013using,shirts2007alchreview}
The close relation between EMUS and MBAR indicates that error analysis of EMUS
may provide insight into the sources of error in MBAR for these problems, but we do not pursue this idea further in the present work.

\section{Numerical Comparison}\label{ssec:numeric_comparison}

To test the algorithm numerically, we performed 100 independent umbrella sampling calculations for the PMF of the $\phi$ coordinate of the alanine dipeptide (i.e., $N$-acetyl-alanyl-$N'$-methylamide) in vacuum.
Simulations were run using GROMACS version 5.1.1 with harmonic bias potentials applied using the PLUMED 2.2.0 software package.\cite{abraham2015gromacs,tribello2014plumed}   The molecule was represented by the CHARMM 27 force field without CMAP corrections,\cite{mackerell2000development} with covalent bonds to hydrogen atoms constrained by the SHAKE algorithm.\cite{ryckaert1977numerical}
Twenty windows were evenly spaced along the $\phi$ dihedral angle.  The force constant $k_i = 0.00760535 \times 10^{-2}$ kcal mol$^{-1}$ degree$^{-2}$ such that the standard deviation of the Gaussian bias functions was 9$^\circ$.
In each window, we integrated the equations of motion with the GROMACS leap-frog Langevin integrator with a 1~fs time step.  The system was equilibrated for 40~ps and then sampled for 100~ps, saving structures every 10~fs.

\begin{figure}
  \includegraphics{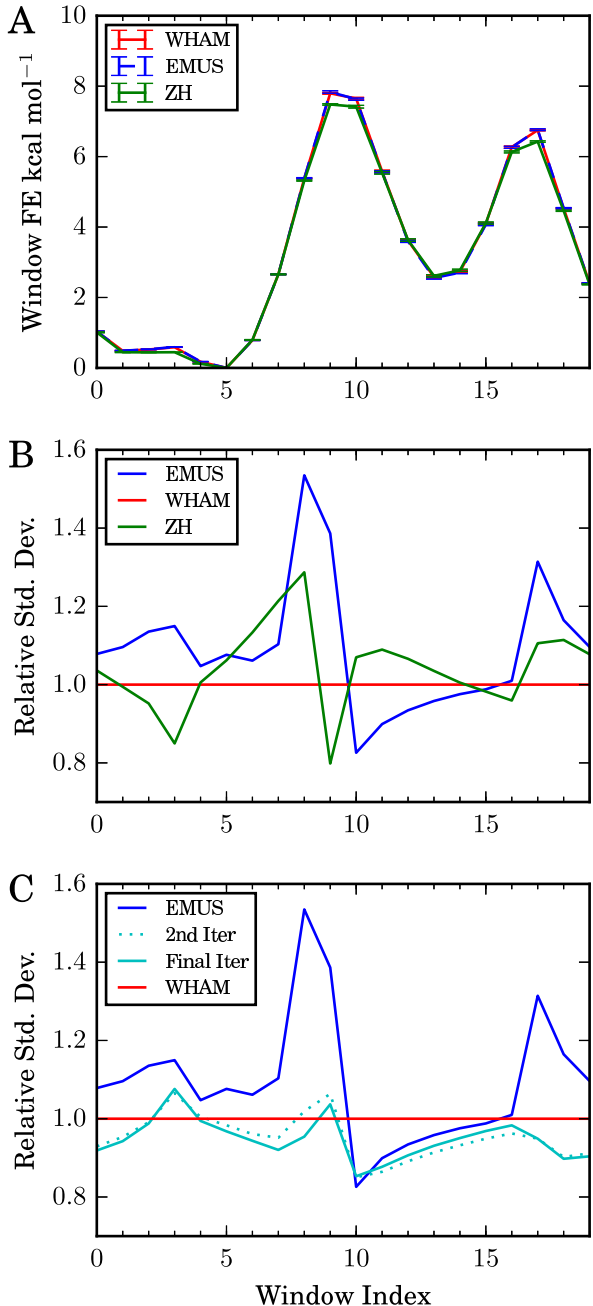}
  \caption{Comparison of umbrella sampling methods applied to simulation data for the alanine dipeptide.  (A) Average window free energies, $G_i$, for the indicated methods.  Error bars are estimated standard deviations of the means.  (B) Standard deviation of each method relative to that of the WHAM algorithm.  Colors are the same as in (A).    (C) EMUS as the first step in a self-consistent iteration to solve the MBAR equations (see text).  The number of uncorrelated samples in each window ($n_i$) was estimated by calculating the integrated autocorrelation of the $\phi$ dihedral angle from each trajectory.  Results shown are for identical molecular dynamics data (see text for simulation details); the methods differ only with respect to combination of the data to estimate the free energies.}
  \label{fig:norm_FE_comp}%
\end{figure}


The data was then analyzed with EMUS, Grossfield's implementation of WHAM,\cite{grossfieldWHAM} and  the algorithm proposed by Zhu and Hummer (ZH) (see equation A1 and following discussion in Appendix A).\cite{zhu2012convergence}  The data was also analyzed with pyMBAR;\cite{shirts2008statistically} as pyMBAR gave results virtually identical to WHAM, the results are not shown.
In Figure \ref{fig:norm_FE_comp}, we show the resulting average potentials of mean force, as well as the standard deviation of the estimates over the 100 runs.
WHAM and EMUS converge to the same result.
This is to be expected, as both algorithms are consistent (i.e., they converge to the exact result as the amount of samples in each window tends to infinity; see Section \ref{sec:analysis}), although WHAM exhibits a small bias from the binning of data for the histograms.\cite{shirts2008statistically}
Although the standard deviations of the free energies are generally higher for EMUS than for WHAM, they are of comparable magnitude.  Compared to the ZH algorithm, EMUS also has a higher standard deviation.  However, ZH is based on thermodynamic integration and uses the trapezoid rule to calculate free energy differences between windows.\cite{zhu2012convergence}  As a consequence it suffers from noticeable quadrature error,\cite{suli2003introduction} which causes the barrier height to converge to an artificially low value.

In Figure \ref{fig:norm_FE_comp}~C, we apply the self-consistent iteration.  For this calculation, we estimate the number of independent samples in each window ($n_i$) from the integrated autocorrelation time of the $\phi$ dihedral angle time series.  We plot the standard deviation of the values of $z$ calculated after the first iteration (EMUS), the second iteration, and after convergence to a relative residual smaller than $10^{-6}$.  In general, convergence is achieved after an average of 9 iterations; none of the 100 data sets required more than 15 iterations.   However, we note that  after two iterations, the estimates of $z$ already have a standard deviation equivalent to that of the WHAM algorithm.

\section{Justification for umbrella sampling by scaling arguments}\label{sec:scaling}

The quality of a statistical estimate from umbrella sampling depends strongly on the choices made for the simulation windows.
In this section, we discuss how the error scales as properties of the simulation change.  We begin in Subsection \ref{sec:flat_scaling} with a description of a prevalent justification for the use of US.  We show in Subsection \ref{sec:flat_scaling2} that this argument is incomplete and, in turn, misleading.   In Subsection \ref{sec:low_T_limit}, we provide an alternative justification; namely, we show that in the low temperature limit, the cost to achieve a fixed accuracy by US grows slowly compared to direct simulation.  In this Section we make several simplifying assumptions that allow us to draw precise conclusions about the scaling properties of EMUS.  In Section \ref{sec:analysis} we provide error bounds for EMUS under much more general assumptions.

\subsection{Scaling in the Limit of Many Windows}\label{sec:flat_scaling}

To justify umbrella sampling, it is often suggested that the total computational time required to accurately sample statistics is inversely proportional to the number of windows, $L.$\cite{chandler1987introduction,frenkel2001understanding,chipot2007free,van1992computer,comer2014adaptive}  The argument for this scaling proceeds as follows.
\begin{itemize}
\item
Divide a one-dimensional collective variable space into $L$ windows of equal length,
inversely proportional to $L$ (i.e., $L^{-1}$).
\item \label{diffusive}
Assume the windows are small enough that
no free energy barriers exist in each window.
The time to explore a window should be diffusion limited and
proportional to the length of the window squared.
Therefore, the simulation time required to accurately sample statistics
in one window  is also proportional to $L^{-2}$.
\item Because there are $L$ windows,
the total simulation time required to compute averages to fixed accuracy should scale as $L \times L^{-2} = L^{-1}$.
\end{itemize}
While this argument is now standard,\cite{chandler1987introduction,frenkel2001understanding,chipot2007free,van1992computer,comer2014adaptive} Virnau and M\"{u}ller \cite{virnau2004calculation} observed that the error for computing the free energy difference between phases of Lennard-Jones particles with an approximately fixed amount of sampling was insensitive to the number of windows in practice, and they noted that the argument above neglects the error associated with combining the data from different simulation windows.
This intuition is supported by our analysis in the next subsection, which shows that the total computational cost to achieve a fixed accuracy should be insensitive to the choice of $L$, so long as it is sufficiently large.

\subsection{A Simple Model Problem}\label{sec:flat_scaling2}
To perform a more precise analysis, we make a number of simplifying assumptions.  We emphasize that these assumptions are in force only for the purposes of the scaling arguments in this section.  We provide more general error bounds for EMUS in Section \ref{sec:analysis}. 

\begin{assumption}[]
The total computation time, $N$, is divided equally among the windows such that $N_i = N/L$.
\end{assumption}

\begin{assumption}[]\label{asm:tridiag}
The $\psi_i$ are functions on the one-dimensional interval $[0,1]$, and  the set of points where $\psi_i$ is non-zero, $\{q:\, \psi_i>0\},$ is an interval of length $\lvert\{q:\, \psi_i>0\}\rvert \leq  \gamma/L.$ We also assume that $\psi_i \psi_j = 0$  unless $|j-i|<0$.    Consequently, both the exact matrix $F$ and the sample mean $\bar F$ are tri-diagonal.
{\rm This assumption clearly does not hold when the bias functions $\psi_i$ are Gaussian.  Nonetheless, the rapid decay of Gaussian bias functions away from their peaks guarantees that entries of $F$ and $\bar F$ far from the diagonal are very small, such that we expect our conclusions to still hold (though their justification would be more complicated).}
\end{assumption}

\begin{assumption}[]\label{asm:tridiagirr}
The overlap of $\psi_i$ and $\psi_{i\pm 1}$  (i.e., the integral of their product) is large enough that
\begin{equation}\label{eq:smoothness}
 \min\{ F_{i,i+1}, F_{i,i-1}\}>  \delta > 0
\end{equation}
for all $L$ and for all $i\leq L$.
{\rm If our last assumption holds, but this one does not, then  we can find more than one vector $z$ satisfying equations \eqref{eq:eval_for_z} and \eqref{eq:normalization of z}.  This assumption is a slightly stronger version of the notion of irreducibility that we defined earlier (see \ Section \ref{sec:eigenproblem}).  Note that we require the irreducibility to hold uniformly in the large $L$ limit, and we thus introduce the $\delta$, which is independent of $L$.}
\end{assumption}

\begin{assumption}[]
Sample averages computed in \emph{different} windows are independent, i.e., $\bar F_{i,i\pm 1}$ and $\bar F_{j,j\pm1}$ for $j\neq i$ are independent.
{\rm We do \emph{not} assume (here or anywhere else in this paper) that samples generated within a single window are independent. 
Indeed, even if the samples from $\pi_i$ are independent, 
$\bar F_{i,i+1}$ and $\bar F_{i,i-1}$ are dependent random variables.}
\end{assumption}

As an example average, let us consider the error in the free energy difference between the first and last windows:
\begin{equation}\label{eq:1D_FE_diff}
\Delta \bar{ G}_{L,1} = -k_BT\ln \left(\frac{ {z}_L}{ {z}_1}\right).
\end{equation}
Assumption~\ref{asm:tridiag} is sufficient for $z(\bar F)$ to be in detailed balance with $\bar{F}$
(Kelly\cite{kelly2011reversibility}, Lemma~1.5 and Section~1.3):
\begin{equation}\label{eq:recursive_form}
{z}_{i+1}(\bar F) =  {z}_{i}(\bar F)\frac{\bar{F}_{i,i+1}}{\bar{F}_{i+1,i}}.
\end{equation}
Using \eqref{eq:recursive_form} recursively,
\begin{align}
\Delta \bar{G}_{L,1} &= -k_BT\ln \left(\frac{\prod_{i=1}^{L-1}\bar{F}_{i,i+1}}{\prod_{i=1}^{L-1} \bar{F}_{i+1,i}}\right)\nonumber \\
&= -k_BT\ln \bar{F}_{1,2} + k_BT\ln \bar{F}_{L,L-1} \nonumber\\
&\quad\quad + k_BT\sum_{i=2}^{L-1} \ln \left ( \frac{\bar{F}_{i,i+1}}{\bar{F}_{i,i-1}} \right ). \label{eq:FE_breakdown}
\end{align}
To understand the error (variance) of the terms in \eqref{eq:FE_breakdown}, we must further specify $F_{i,i+1}$ and $F_{i,i-1}$.

\begin{assumption}[]\label{asm:difflim}
For $N_{min}$ and $L_{min}$ sufficiently large, when $N_i \geq N_{min}$ and $L\geq L_{min}$,
\begin{equation}\label{difflim}
\frac{K_{min}}{N_i L^2} \leq  \var \left( \ln \left(\frac{ \bar F_{i,i+ 1}}{\bar F_{i,i-1}}\right)\right) \leq  \frac{K_{max}}{N_i L^2}
\end{equation}
for $i=2,3,\dots, L-1$, and the same upper and lower bounds hold for $\var\left( \ln  F_{1,2}\right)$ and $\var\left( \ln F_{L,L-1}\right)$. {\rm
This is just a precise interpretation of the diffusion limited sampling assumption made in the standard justification of US reproduced in the last subsection.  
Under such an assumption we expect both $\bar F_{i,i+1}$ and $\bar F_{i,i-1}$ to have variance on the order of $1/(N_i L^2)$ and, in light of \eqref{eq:smoothness}, the function $\ln (x/y)$ is smooth near $(x,y) = (F_{i,i+1}, F_{i,i-1})$.   These considerations are closely related to Lemma \ref{lem:delta method} in the next section.}
\end{assumption}


With all the assumptions in hand, we now complete the argument by taking the variance of both sides of \eqref{eq:FE_breakdown}.
Since samples from different windows are independent,
the variance of $\bar{G}_{L,1}$ is a sum of contributions from each window:
\begin{equation}\label{eq:FE_var}
\begin{split}
\var \left( \Delta \bar{G}_{L,1}\right) &= \var \left(k_BT \ln \bar{F}_{1,2} \right) + \var \left(k_BT \ln\bar{F}_{L,L-1} \right) \\
& \qquad +\sum_{i=2}^{L-1} \var \left( k_BT\ln \left ( \frac{\bar{F}_{i,i+1}}{\bar{F}_{i,i-1}} \right ) \right).
\end{split}\
\end{equation}
Using \eqref{difflim} and substituting $N_i=N/L$, we find that, as  long as
$N/L\geq N_{min}$ and $L\geq L_{min},$
\begin{equation}\label{flaterrscale}
(k_B T)^2 \frac{K_{min}}{N}\leq \var \left( \Delta \bar{G}_{L,1}\right) \leq   (k_B T)^2 \frac{K_{max}}{N}.
\end{equation}
We thus see that the variance is independent of $L$.

To verify that this conclusion carries over to harmonic bias potentials, we performed multiple umbrella sampling calculations for a Brownian particle on a flat
potential on the interval $\left[0,1\right]$ with a stepsize of $1.0\times 10^{-6}$ and $k_B T=1.0$ using Gaussian bias functions with a standard deviation of $1/L$.  The number of windows was varied from $L=10$ to 46 in steps of 2.  For each value, a total of $10^7$ steps were distributed equally in the windows and the US calculation was repeated 480 times.  We then calculated the mean square error of the free energy difference between the first and last window over the 480 replicates and determined how the mean square error scaled with $L$.  An $L^{-1}$ scaling would predict that mean square error would decrease inversely with the number of windows used.  By contrast,
the data plotted in Figure \ref{fig:flatscaling} support a scaling of $L^0$, consistent with \eqref{flaterrscale}.

It is worth noting that the inverse scaling with total cost $N$ in \eqref{flaterrscale} is exactly the scaling one would expect for the variance of an estimate of the free energy difference $\Delta \bar{G}_{L,1}$ constructed from a molecular dynamics trajectory of length $N$.  Because US and direct simulations of comparable total numbers of steps require comparable computational effort (ignoring the overhead associated with combining the simulation data, which is typically small in comparison with the computational cost of the sampling), the benefits of US must be encoded in the constants $K_{min}$ and $K_{max}.$   A dramatic demonstration of this observation is the purpose of the next subsection.

\begin{figure}
  \includegraphics[scale=1.0]{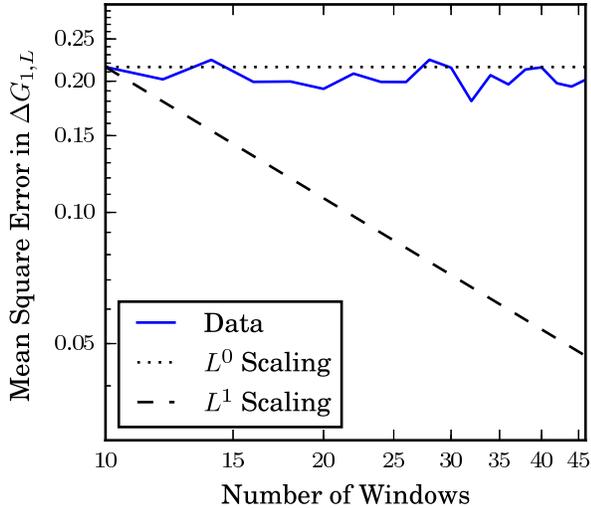}
  \caption{The scaling of umbrella sampling error with number of windows on a flat potential.  A Brownian particle on a flat, one-dimensional potential was simulated for 480 identical runs, and the free energy difference between the first and last windows was calculated, as described in the text.  Here, the mean square error from the exact result is plotted against the number of windows.   The lines show the scaling in error predicted by the $L^{-1}$ and $L^0$ scalings.  Fitting the data on a log-log scale gives a scaling exponent of $-0.026\pm0.028$.}
  \label{fig:flatscaling}
\end{figure}

\subsection{The Low Temperature Limit}\label{sec:low_T_limit}
To understand the benefits of umbrella sampling,
we must study its performance when free energy barriers are large in comparison with
the temperature, i.e., the low temperature limit ($k_BT\rightarrow 0$).
In this limit, the cost of direct sampling increases exponentially with $1/T$, while, as we show, the cost of umbrella sampling increases only algebraically.
A formal discussion is given in a separate publication;\cite{VanKotenArxiv}
here, we present a simple plausibility argument.

Owing to the free energy barriers, the assumption of diffusive dynamics in each window no longer holds.
Instead, we expect a form typical of reaction rate theories in each window.
We define $\Delta W_i$ as the maximum difference in the PMF in window $i$:
\begin{equation}\label{eq:defn_Delta_H_i}
\Delta W_i = \max_{\{q:\, \psi_i(q)>0\}} \left\{W\left(q\right)\right\} - \min_{\{q:\, \psi_i(q)>0\}} \left\{ W\left(q\right)\right\}.
\end{equation}

\noindent {\bf  Assumption \ref{asm:difflim}$'$.}
{\it We now replace the upper and lower bounds in \eqref{difflim} by the upper bound
\begin{equation}\label{arrhenius_kinnetics}
\var \left( \ln \left(\frac{ \bar F_{i,i+ 1}}{\bar F_{i,i-1}}\right)\right) \leq  \frac{K}{k_BTN_i L^2} \exp\left({\frac{\Delta W_i }{k_B T}}\right)
\end{equation}
for $i=2,3,\dots,L-1$ with analogous replacements for $i=1$ and $i=L$, as long as
$N_i\geq N_{min}$ and $L\geq L_{min}.$   The constant $K$ here is assumed to be independent of temperature.}
This bound captures the diffusion limited sampling assumption when $L$ is very large, but is more detailed than \eqref{difflim} in that it captures (crudely) the increasing difficulty of the sampling problem as the temperature decreases with all other parameters held fixed.   Under reasonable additional assumptions on the underlying potential, the bias functions $\psi_i,$  and the sampling scheme, one can rigorously establish an asymptotic (large $N_i$) bound of the form in \eqref{arrhenius_kinnetics}.\cite{VanKotenArxiv}

Substituting this new bound into \eqref{eq:FE_var}, we find that, if $L\geq L_{min}$ and $N/L\geq N_{min},$  then
\begin{equation}\label{dGbnd1}
\var \left( \Delta \bar{G}_{L,1} \right) \leq  \frac{K k_B T}{NL} \sum_{i=1}^L   \exp \left(\frac{\Delta W_i}{k_B T}\right).
\end{equation}
As the temperature decreases, we choose to increase $L$ such that $\Delta W_i / k_B T$ is bounded above.
This can be achieved by scaling $L$ linearly with $1/T$: if the derivative of the PMF is bounded (in absolute value) by $W_{max}'$, choosing $L$ so that
\begin{equation}\label{Lbigenough}
L \geq \frac{W_{max}' }{\Omega k_B T}
\end{equation}
ensures that $\Delta W_i / k_B T$ is bounded by $\Omega$ (since we have assumed that the argument of $W$ is in $[0,1]$).
On the other hand, our assumption that the length of $\{q:\, \psi_i>0\}$ (Assumption \ref{asm:tridiag})  does not exceed $\gamma/L$ implies that \begin{equation}
\frac{\Delta W_i}{k_B T} \leq \frac{W_{max}' \gamma }{k_B T L}.
\end{equation}
Consequently, as long as \eqref{Lbigenough} holds,
\[
\frac{\Delta W_i}{k_B T} \leq \Omega\gamma.
\]
Finally, substituting this result into \eqref{dGbnd1} we find that  if $L\geq L_{min}$ and $N/L\geq N_{min},$ then
\begin{equation*}
\var \left( \Delta \bar{G}_{L,1} \right) \leq \frac{ K k_B T\exp(\Omega\gamma)}{N}
\leq \frac{ K L k_B T\exp(\Omega\gamma)}{N_{min}}.
\end{equation*}
With the best possible (smallest) choice of $L$ allowed by \eqref{Lbigenough}, this bound becomes
\begin{equation}\label{eq:lowTresult}
\var \left( \Delta \bar{G}_{L,1} \right) \leq \frac{ K W'_{max} \exp(\Omega\gamma)}{\Omega N_{min}}.
\end{equation}
The remarkable feature of the bound in \eqref{eq:lowTresult} is that it is independent of $T$.
This does not mean that the cost to achieve a fixed accuracy is independent of $T.$    However, it does imply that as the temperature is decreased we do not have to increase $N_{min}$ to maintain a fixed accuracy.   Expression \eqref{Lbigenough} and the fact that $N_i \geq N_{min}$  together imply that the computational cost of obtaining an accurate estimate of $\Delta \bar{G}_{L,1}$ by US increases algebraically with $(k_B T)^{-1}$.  That scaling is to be compared to exponential in $(k_B T)^{-1}$ to achieve the same accuracy by direct simulation.


\section{Analysis of the Error of EMUS}\label{sec:analysis}

In this section, we study the error of EMUS in full generality,
without imposing the simplifying assumptions of the previous section.
Our main results are a central limit theorem for EMUS (Theorem \ref{thm:clt for emus} below)
and an easily computed, practical error estimator which reveals the contributions
of the different windows to the total error.
These results may be used to compare the efficiency of EMUS and other methods
and to study how the efficiency of EMUS depends on parameters such as the
number of samples allocated to each window.

\subsection{A Central Limit Theorem for EMUS}\label{sec:clt}

Before developing the error analysis,
we define a single notation for EMUS which incorporates
both the case of a free-energy difference and the case of an ensemble average.
In either case, one must compute $\bar F$
and also $\bar g_{1,i}^\ast$ and $\bar g_{2,i}^\ast$
for two real valued functions $g_1$ and $g_2$.
To compute a free energy difference, we choose based on \eqref{eq:defn_FE}
\begin{equation*}
  g_1 = \1_{S1} \text{ and } g_2 = \1_{S2}.
\end{equation*}
To compute an ensemble average $\langle g\rangle$, we choose based on \eqref{eq:combined_avg}
\begin{equation*}
  g_1 = g \text{ and } g_2 = 1.
\end{equation*}
We furthermore define the function
\begin{equation}\label{eq:defn_chi_fxn}
  \avevec_i(x) = \left (\psi_1^\ast(x), \dots, \psi_{L}^\ast(x), g_{1,i}^\ast(x), g_{2,i}^\ast(x) \right )
\end{equation}
so that
\begin{equation*}
  \bar \avevec_i = \frac{1}{N_i} \sum_{t=0}^{N_i-1} \avevec_i(X^i_t)
  = \left(\bar{F}_{i1},\dots, \bar{F}_{iL}, \bar g_{1,i}^\ast, \bar g_{2,i}^\ast \right),
\end{equation*}
where we remind the reader that, for each $i,$ the process $X^i_t$ samples the biased distribution $\pi_i$.
Define
\begin{equation*}
  \bar \avevec = (\bar \avevec_1, \dots, \bar \avevec_L),
\end{equation*}
and let
\begin{equation*}
 \<\avevec\> = \left(\<\avevec_1\>_1, \dots, \<\avevec_L\>_L\right)
\end{equation*}
denote the corresponding vector of exact averages.
The EMUS estimator takes the form $B(\bar \avevec)$,
where for a free-energy difference,
\begin{equation}\label{eq: B for a free energy difference}
  B(\bar \avevec) = -kT \log \left (
  \frac{\sum_{i=1}^L z_i(\bar F) \bar g_{1,i}^\ast}{\sum_{i=1}^L z_i(\bar F) \bar g_{2,i}^\ast}
  \right ),
\end{equation}
and for an ensemble average,
\begin{equation}\label{eq: B for an ensemble average}
  B(\bar \avevec) = \frac{\sum_{i=1}^L z_i(\bar F) \bar g_{1,i}^\ast}{\sum_{i=1}^L z_i(\bar F) \bar g_{2,i}^\ast}.
\end{equation}

We now proceed with the error analysis.
First, we characterize the error of the sample means
over the biased distributions.
As discussed by Frenkel and Smit~\cite[Appendix~D]{frenkel2001understanding},
the variance of a sample mean may be expanded in terms of the integrated
autocovariance of the process.
We define the autocovariance function of $\avevec_i(X^i_t)$ to be
\begin{equation*}
  C_i(t) = \left\langle (\avevec_i(X^i_0) - \<\avevec_i\>_i)(\avevec_i(X^i_t) - \<\avevec_i\>_i)^\t\right\rangle_i,
\end{equation*}
where $\t$ denotes a vector transpose, and here the outer $\langle\ldots\rangle_i$ denotes the exact average not only over $X^i_0$ sampled from $\pi_i$ but also subsequent points of the sequence $X^i_t$.
Note that $C_i(t)$ is a $(L+2) \times (L+2)$  matrix.
We define the integrated autocovariance to be
\begin{equation}\label{eq:integrated autocovariance}
  \Sigma_i = \sum_{t=-\infty}^\infty C_i(t).
\end{equation}
The integrated autocovariance is the leading
order coefficient in an expansion of the covariance $\bar \avevec_i$ (see Frenkel and Smit\cite[D.1.3]{frenkel2001understanding}):
\begin{equation}\label{eq: expansion of covariance}
  \cov(\bar \avevec_i) = \frac{\Sigma_i}{N_i} + o \left ( \frac{1}{N_i} \right ),
\end{equation}
where $o(1/N_i)$ denotes terms that go to zero faster than $N_i$ (i.e., $N_io(1/N_i)\rightarrow 0$).

Under certain conditions on the process $X^i_t$,
one can strengthen the expansion of the covariance \eqref{eq: expansion of covariance} to a central limit theorem (CLT) for $\bar \avevec_i$.
We expect such a CLT to hold for most problems and most sampling methods
in computational statistical physics.
However, to avoid a lengthy and technical digression, we simply take the CLT as an assumption;
we justify this assumption in more detail in another work,\cite{VanKotenArxiv}
and we refer to Leli\`{e}vre {\it et al.}\cite[Section~2.3.1.2]{LeLievre:FreeEnergyComp}
for a general discussion of the CLT
in the context of computational statistical physics.

\begin{assumption}[Central Limit Theorem for $\bar \avevec_i$]\label{asm:clt for sample means}
We assume that
\begin{equation}\label{eq:assumption CLT for window averages}
\sqrt{N_i}\left(\bar \avevec_i -\<\avevec_i\> \right) \convdist \N \left(0,\Sigma_i\right)
\end{equation}
where ${\Sigma_i \in \Real^{(L+2)\times(L+2)}}$ is the integrated autocovariance matrix
defined in~\eqref{eq:integrated autocovariance}.
The symbol $\convdist$ denotes convergence in distribution as $N_i \rightarrow \infty$.
{\rm Notice that when the elements of the sequence $X_t^i$ are independent and drawn from $\pi_i$ then
$\Sigma_i =  \left\langle (\avevec_i - \<\avevec_i\>_i)(\avevec_i - \<\avevec_i\>_i)^\t\right\rangle_i/N_i$.  More generally, samples are correlated, so $\Sigma_i$ includes a factor that accounts for the time to decorrelate.}
\end{assumption}

Having characterized the errors in the sample means, we now study how these errors
propagate through the EMUS algorithm.
Our goal is to prove a CLT for EMUS.
We accomplish this using the delta method.

\begin{lemma}[The Delta Method; Proposition~6.2 of Bilodeau and Brenner~\cite{bilodeau2008theory}]\label{lem:delta method}
Let $\theta_N$ be a sequence of random variables taking values in $\Real^d$.
Assume that a central limit theorem holds for $\theta_N$
with mean $\mu \in \Real^d$ and asymptotic covariance matrix $\Sigma \in \Real^{d \times d}$;
that is, assume
\begin{equation*}
  \sqrt{N} (\theta_N - \mu) \convdist \N(0, \Sigma).
\end{equation*}
Let $\Phi:\Real^d \to \Real$ be a function differentiable at $\mu$.
Then we have the central limit theorem
\begin{equation*}
\sqrt{N} \left(\Phi \left(\theta_N\right) - \Phi \left(\mu\right) \right) \convdist \N \left(0,\nabla \Phi^\t \left(\mu\right)  \Sigma \nabla \Phi \left(\mu\right) \right)
\end{equation*}
for the sequence of random variables $\Phi(\theta_N)$.
\end{lemma}

To motivate the delta method, we observe that if $X$ has distribution $\N(\mu, \Sigma)$,
then $\nabla \Phi(\mu)^\t X$ has distribution
$\N \left(\Phi(\mu),\nabla \Phi^\t \left(\mu\right) \Sigma \nabla \Phi \left(\mu\right) \right)$.
That is, according to the delta method, the asymptotic distribution of $\Phi(X)$
is the linearization of $\Phi$ at $\mu$
applied to the asymptotic distribution of $X$.
Thus, one may regard the delta method as a rigorous version of
the standard error propagation formula based on linearization.

We prove the CLT for EMUS by applying the delta method
with $\bar \avevec$ taking the place of $\theta_N$ and
with the function $B$ taking the place of $\Phi$.
We require the following assumptions in addition to
Assumption~\ref{asm:clt for sample means}.
\begin{assumption}\label{asm: assumptions in addition to CLT for smpl means}
We assume:
\begin{enumerate}
\item
The proportion of the total number of samples drawn from each window is constant
in the limit as $N \rightarrow \infty$; that is,
\begin{equation}\label{eq:assumption of fixed proportions of samples}
  \lim_{N \rightarrow \infty} N_i/N = \kappa_i.
\end{equation}
\item Sampling in different windows is independent; that is, $\bar \avevec_i$ is independent of $\bar \avevec_j$ when $j \neq i$.
\item The biasing functions $\psi_i$ are chosen so that $F$ is irreducible; see Section~\ref{sec:eigenproblem}.
\end{enumerate}
\end{assumption}

We now give the CLT for EMUS.

\begin{theorem}[Central Limit Theorem for EMUS]\label{thm:clt for emus}
Let Assumptions~\ref{asm:clt for sample means}
and~\ref{asm: assumptions in addition to CLT for smpl means} hold.
Let
\begin{equation*}
\frac{\partial B}{\partial \bar \avevec_i} = \left ( \frac{\partial B}{\partial \bar F_{i1}},
 \dots, \frac{\partial B}{\partial \bar F_{iL}}, \frac{\partial B}{\partial \bar g^\ast_{1,i}},
 \frac{\partial B}{\partial \bar g^\ast_{2,i}}  \right ) \in \Real^{L+2}
\end{equation*}
denote the partial derivative of $B$ with respect to $\bar \avevec_i$,
evaluated at $\avevec$.
Under the assumptions stated above,
\begin{equation*}
  \sqrt{N} \left (B\left(\bar \avevec\right) - B\left(\<\avevec\>\right)\right )
  \convdist \N(0, \sigma^2),
\end{equation*}
where
\begin{equation}\label{eq:formula for asymptotic variance}
  \sigma^2 = \sum_{i=1}^L \frac{1}{\kappa_i} \left(\frac{\partial B}{\partial \avevec_i}^\t  \Sigma_i \frac{\partial B}{\partial \avevec_i}\right).
\end{equation}
We refer to $\sigma^2$ as the asymptotic variance of EMUS.
\end{theorem}
\begin{proof}
First, we write down a central limit theorem for
$(\bar \avevec_1, \dots, \bar \avevec_L)$.
We have that
\begin{equation}\label{eq:adjusted clt for windows}
  \sqrt{N} (\bar \avevec_i - \<\avevec_i\>_i ) \convdist \N (0, \kappa_i^{-1} \Sigma_i)
\end{equation}
by Assumption~\ref{asm:clt for sample means}
and~\eqref{eq:assumption of fixed proportions of samples}.
Since the sampling in different windows is assumed to be independent,
\eqref{eq:adjusted clt for windows} implies
\begin{equation*}
  \sqrt{N} (\bar \avevec  - \<\avevec\>) \convdist \N (0, \Sigma),
\end{equation*}
where $\Sigma \in \Real^{L(L+2) \times L(L+2)}$ is the block diagonal matrix
\begin{equation*}
\Sigma =\left[
\begin{array}{c|c|c|c}
\Sigma_1 / \kappa_1 & 0 & 0 &  \ldots \\ \hline
0 & \Sigma_2 / \kappa_2 & 0 &  \ldots \\ \hline
0 & 0 & \Sigma_3 / \kappa_3 & \ldots \\ \hline
\vdots & \vdots & \vdots & \ddots \\
\end{array}
\right].
\end{equation*}

Second, we verify that $B$ is differentiable at $\bar{\avevec}$.
Since $F$ is assumed to be an irreducible stochastic matrix,
$z(\bar F)$ is differentiable at $\bar{F}$.
We refer to  Thiede {\it et al.}\cite{perturbationArticle} Lemma~3.1  for a complete explanation.
It follows from the chain rule that $B$ is differentiable at $\bar{\avevec}$.

Finally, applying Lemma~\ref{lem:delta method} with $B$ playing the role of $\Phi$
and $\bar \avevec$ the role of $\theta_N$ concludes the proof.
\end{proof}

The asymptotic variance $\sigma^2$ appearing in Theorem~\ref{thm:clt for emus}
measures the rate at which the error of EMUS decreases with the number of samples.
To make this precise, we observe that Theorem~\ref{thm:clt for emus}
is equivalent to the following asymptotic result concerning confidence intervals.
For every $\alpha >0$,
\begin{equation}
  \lim_{N\rightarrow \infty} \P \left [
  \left \lvert B(\bar \avevec) - B(\<\avevec\>) \right \rvert  \leq \frac{\alpha \sigma}{\sqrt{N}}  \right ]
   = \erf \left ( \frac{\alpha}{\sqrt{2}} \right ),
\end{equation}
where $\bf P$ denotes a probability and $\erf$ denotes the error function.

The asymptotic variance is commonly used to measure the efficiency of an estimator.
We refer to van der Vaart~\cite{vanderVaart:AsymptStat}
for an explanation and for a discussion of other possibilities.
In Section~\ref{sec:estimate of asymptotic var},
we explain how the proportion $\kappa_i$ of samples allocated to each
window may be adjusted to minimize the asymptotic variance of EMUS,
thereby maximizing efficiency.

We note that a central limit theorem similar to Theorem~\ref{thm:clt for emus}
has been proved for the MBAR estimator by Gill {\it et al.}~\cite[Proposition~2.2]{gill1988}.
However, the authors of this work do not study the dependence of the asymptotic variance on the parameters, as we do.
In fact, the MBAR estimator is significantly more complicated than EMUS,
and its dependence on the number of windows and the allocation of samples is harder to understand.

Van Koten and Weare\cite{VanKotenArxiv} use a result similar to Theorem \ref{thm:clt for emus} to generalize the conclusions of Section \ref{sec:scaling} to periodic and multi-dimensional reaction coordinates and to a wider class of observables than free energy differences. We show both that the asymptotic variance is constant in the limit of large $L$ and that the work required to compute an average to fixed precision increases only algebraically in the low temperature limit. In addition, we use recently developed perturbation estimates for Markov chains \cite{perturbationArticle} to quantify the dependence of the asymptotic variance of EMUS on the degree to which the bias functions overlap.

\subsection{Estimating the Asymptotic Variance of EMUS}\label{sec:estimate of asymptotic var}

Our goal in this section is to derive a computable estimate $\bar{\sigma}^2$
of the asymptotic variance $\sigma^2$,
which can be decomposed to assess the contributions from individual windows to errors in averages.
We recall that formula~\eqref{eq:formula for asymptotic variance} for $\sigma^2$
involves partial derivatives of $B$.
Our estimate $\bar{\sigma}^2$ of $\sigma^2$ requires explicit formulas for these partial derivatives.
We provide the appropriate expressions, both for ensemble averages and for free-energy differences,
in Lemma \ref{lem:formulas for partials of emus}.
Following the partial derivatives,
we present an algorithm for evaluating $\bar{\sigma}^2$ and demonstrate it for the alanine dipeptide.
Finally, we compare with the output of a procedure
from Zhu and Hummer (ZH)\cite{zhu2012convergence} in Section \ref{secn:zhuhummer}.

\begin{lemma}\label{lem:formulas for partials of emus}

We have the following formulas for $\partial B/\partial \bar \avevec_i$:
\begin{enumerate}
\item When EMUS is used to compute an ensemble average $\langle f \rangle$,
$B$ is defined by~\eqref{eq: B for an ensemble average},
and we have
\begin{align*}
  \frac{\partial B}{\partial \bar F_{ij}} (\bar \avevec)
   &= \frac{\sum_{k} z_i(\bar F) (I-\bar F)^\#_{jk} ( \bar g^\ast_k - B(\bar \avevec) \bar 1^\ast_k )}{\sum_k z_k(\bar F) \bar 1^\ast_k}, \\
  \frac{\partial B}{\partial \bar g^\ast_{1,i}}(\bar \avevec) &=
  \frac{z_i(\bar F)}{\sum_k z_k(\bar F) \bar 1_k^\ast}, \text{ and } \\
  \frac{\partial B}{\partial \bar g^\ast_{2,i}}(\bar \avevec)
  &= - \frac{B(\bar \avevec) z_i(\bar F) }{\sum_k z_k(\bar F) \bar 1_k^\ast}.
\end{align*}
\item
When EMUS is used to compute a free-energy difference,
$B$ is defined by~\eqref{eq: B for a free energy difference},
and we have
\begin{align*}
  \frac{\partial B}{\partial \bar F_{ij}} (\bar \avevec) &=
  kT z_i(\bar F) \left (\frac{\sum_k (I-F)^\#_{jk} \1^\ast_{S2,k}}{\sum_k z_k(\bar F) \1^\ast_{S2,k}}
    \right . \\
  &\qquad \qquad \qquad \left . - \frac{\sum_k (I-F)^\#_{jk} \1^\ast_{S1,k}}{\sum_k z_k(\bar F) \1^\ast_{S1,k}} \right ), \\
  \frac{\partial B}{\partial \bar g^\ast_{1,i}} (\bar \avevec) &= kT \frac{z_i(\bar F)}{\sum_k z_k(\bar F) \bar \1_{S1,k}^\ast},
  \text{ and } \\
  \frac{\partial B}{\partial \bar g^\ast_{2,i}}  (\bar \avevec) &= -kT  \frac{z_i(\bar F)}{\sum_k z_k(\bar F) \bar \1_{S2,k}^\ast}.
\end{align*}
\item{
When EMUS is used to compute the free energy of a window, $B$ is defined by \eqref{eqn:windowfreeenergy}, and we have
\begin{align*}
\frac{\partial B}{\partial \bar F_{ij}}  (\bar \avevec)&=
  \frac{z_i(F)}{z_k(F)} (I-F)^\#_{jk}, \text{and}\\
\frac{\partial B}{\partial \bar g^\ast_{i,1}} (\bar \avevec)
&=\frac{\partial B}{\partial \bar g^\ast_{i,2}}  (\bar \avevec)= 0.
\end{align*}
Note that for a free energy difference between windows, we can simply subtract derivatives for the corresponding windows.}
\end{enumerate}
\end{lemma}
\begin{proof}
We begin by reminding the reader that the output of EMUS is the vector of window normalization constants, $z$, which depends on the sample mean $\bar{F}$.
Because all other averages and, in turn, their derivatives rely on $z$, we need to determine the sensitivity of each element of $z$ to each element of $\bar{F}$ (i.e., $\partial z_k/\partial \bar F_{ij}$).
Since $\bar F$ is a stochastic matrix, some care must be taken in defining this derivative.
We resolve the technical difficulties in detail elsewhere;
see Van Koten and Weare~\cite{VanKotenArxiv} and Thiede\textit{ et al.}~\cite[Lemma~3.1]{perturbationArticle}.
Here, to obtain the derivative
$\partial z_k/\partial \bar F_{ij}$, evaluated at $\bar F$, we perturb around $\bar F$:
\begin{equation}
  \left . \frac{d}{d\eps} \right |_{\eps =0} z_k(\bar F + \eps E)
  = \sum_{i,j=1}^L \frac{\partial z_k}{\partial \bar F_{ij}}(\bar F) E_{ij},
  \label{eq:defn of partials}
\end{equation}
where $E$ is an arbitrary matrix, $\eps$ is a scalar, and we assume that the sum $\bar F + \eps E$ is also a stochastic matrix.
The righthand side follows from the chain rule, effectively treating each element of the matrix as a separate argument to each element $z$.  Then, we employ a relation from Golub and Meyer\cite[Theorem~3.1]{golub1986using}:
\begin{equation}\label{eq:golubmayer}
\left . \frac{d}{d\eps} \right |_{\eps =0} z_k(\bar F + \eps E)  = z(\bar F)^\t E (I-\bar F)^\# e_k,
\end{equation}
where \# denotes the group inverse, a generalized matrix inverse similar to the Moore-Penrose inverse.  It is defined as satisfying $AA^\#A = A$, $A^\#AA^\# = A^\#$, $AA^\# = A^\#A$.
We refer to Golub and Meyer\ \cite{golub1986using} for further discussion of the group inverse
and an algorithm for computing it.
Finally, we equate \eqref{eq:defn of partials} and \eqref{eq:golubmayer} and solve for the derivative of interest:
\begin{equation}\label{eq:zsensitivity}
  \frac{\partial z_k}{\partial \bar F_{ij}}(\bar F) = z_i(\bar F) (I-\bar F)^\#_{jk}.
\end{equation}
Thus the sensitivity  of each element of $z$ to each element of $\bar{F}$ can be computed from linear algebra operations.

With \eqref{eq:zsensitivity}, we can now compute the derivatives of $B$.
We derive the formulas for the free-energy difference explicitly; the other cases are similar.  In this case,
\begin{equation*}
  B(\bar \avevec) =
   kT \log \left ( \sum_{k=1}^L z_k(\bar F) \bar g_{2,k}^\ast \right ) -
   kT \log \left (\sum_{k=1}^L z_k(\bar F) \bar g_{1,k}^\ast \right ).
\end{equation*}
By the chain rule,
\begin{align*}
  \frac{\partial}{\partial \bar g_{1,i}^\ast} \log \left (\sum_{k=1}^L z_k(\bar F) \bar g_{1,k}^\ast \right )
  &= \frac{z_i(\bar F)}{\sum_{k=1}^L z_k(\bar F) \bar g_{1,k}^\ast},
\end{align*}
and
\begin{align*}
  \frac{\partial}{\partial \bar F_{ij}} \log \left (\sum_{k=1}^L z_k(\bar F) \bar g_{1,k}^\ast \right )
  &=  \frac{\sum_{k=1}^L \frac{\partial z_k}{\partial \bar F_{ij}}(\bar F) \bar g_{1,k}^\ast}{\sum_{k=1}^L z_k(\bar F) \bar g_{1,k}^\ast}.
\end{align*}
The stated result follows by substituting $g_1 = \1_{S1}$, $g_2 = \1_{S2}$, and
the expression in \eqref{eq:zsensitivity} for $\partial z_k/\partial \bar F_{ij}$.
\end{proof}

\subsubsection{Computational Procedure}\label{secn:algorithm}

We now provide a practical procedure that uses the derivatives above to estimate $\sigma^2$  from
trajectories that sample the distributions $\pi_i$.
For clarity, we assume that the system is equilibrated (i.e., $X^i_0$ has distribution $\pi_i$,
so that the process $X^i_t$ is stationary) throughout this section.

We begin by rewriting \eqref{eq:formula for asymptotic variance} as
\begin{equation}\label{eq:asympt var emus in terms of M processes}
  \sigma^2 = \sum_{i=1}^L \frac{\chi_i^2}{\kappa_i}
\end{equation}
where
\begin{align*}
 \chi_i^2
&= \frac{\partial B}{\partial \avevec_i}^\t \Sigma_i
\frac{\partial B}{\partial \avevec_i} \\
&=
\frac{\partial B}{\partial \avevec_i}^\t \left (\sum_{t=-\infty}^\infty  C_i(t) \right)
\frac{\partial B}{\partial \avevec_i}.
 \end{align*}
 Defining the sequence
 \begin{equation}\label{eq:defn_EMUS_err_process}
  \zeta^i_t = \left.\frac{\partial B}{\partial \avevec_i}\right|_{\<\avevec\>}\cdot \left(\avevec_i(X^i_t) - \<\avevec_i\>_i\right)
\end{equation}
we find that
 \[
 \chi_i^2
 = \sum_{t=-\infty}^\infty \left\langle \zeta^i_t \zeta^i_0 \right\rangle
 \]
 which is the integrated autocovariance of $\zeta_t^i.$

We thus propose the following algorithm, given simulation data:
\begin{enumerate}
  \item Compute $\bar \avevec$.
  \item Compute $z(\bar F)$ and $(I-\bar F)^\#$ using the algorithm of Golub and Meyer.\cite{golub1986using}
  \item Evaluate $\partial B/\partial \avevec_i$  at $\bar \avevec$
  using the formulas in Lemma~\ref{lem:formulas for partials of emus}.
  \item Compute
  \begin{equation*}
    \bar \zeta^i_t = \left.\frac{\partial B}{\partial \avevec_i}\right|_{\bar \avevec_i} \cdot (\avevec_i(X^i_t) - \bar \avevec_i).
  \end{equation*}
  \item Compute an estimate $\bar \chi_i^2$ of the integrated autocovariance of
  $\bar \zeta^i_t$ using an algorithm such as ACOR.\cite{acor}
  \item Compute the estimate of $\sigma^2$:
  \begin{equation}\label{eq:approx formula for asymptotic variance}
    \bar \sigma^2 = \sum_{i=1}^L \frac{\bar \chi_i^2}{\kappa_i}.
  \end{equation}
\end{enumerate}
Since $\bar F$, $\bar \avevec$, and  $z$ are all computed in the process of obtaining the EMUS averages, estimating $\bar \sigma^2$ only requires one additional pass over the simulation data.  This additional cost is insignificant compared with that of computing the trajectories.

Both \eqref{eq:formula for asymptotic variance}
and its approximation~\eqref{eq:approx formula for asymptotic variance}
decompose the asymptotic variance of EMUS into a sum of contributions from each window.
By comparing the sizes of terms in the sum, we can determine the degrees to which different windows contribute to the error.
In principle, this information can be used to guide modification of the parameters of the simulation to improve efficiency.
For instance, one might adjust the proportion of samples allocated to each window, $\kappa_i$, to minimize the asymptotic variance.
From \eqref{eq:asympt var emus in terms of M processes}, the asymptotic variance $\sigma^2$ is minimized when $\kappa_i \propto\chi_i$ (see \ Zhu and Hummer\cite[42]{zhu2012convergence}).
Consequently, we can define the \textit{relative importance} of window $i$ as
\begin{equation}\label{eqn:relative importance}
\mu_i = L\frac{\chi_i}{\sum_{k=1}^L \chi_i},
\end{equation}
where the normalization is chosen so that $\mu_i=1$, regardless of $L$, if all windows have the same importance.
The relative importance represents how many samples would be allocated to a window to optimally estimate a specific observable, compared to a uniform distribution over all umbrellas.

\subsubsection{Numerical Results}\label{secn:alanineerrors}

\begin{figure}
  \includegraphics{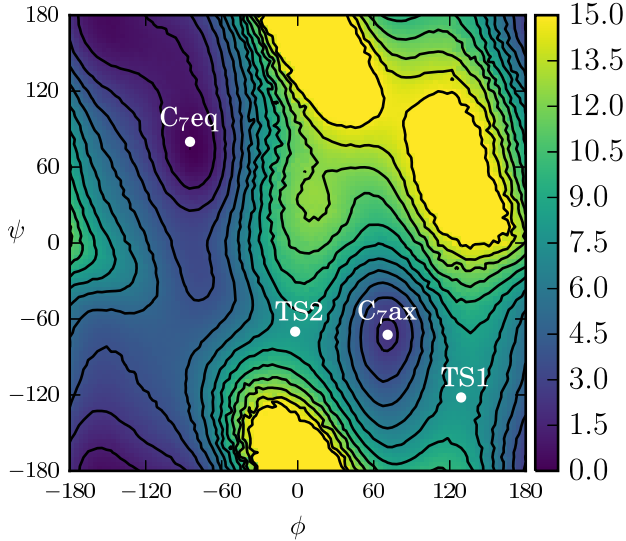}
  \caption{Potential of mean force obtained from US with biases on the $\phi$ and $\psi$ dihedral angles.  Major basins and barriers on pathways connecting them are indicated.  The scale bar indicates PMF values in kcal/mol, and the contour spacing is 2 $k_B T$.   The surface is constructed from simulation data accumulated in histograms with 100 bins in each collective variable.  See text for simulation details.}
   \label{fig:2D_umb_fes}
\end{figure}


To study the behavior of these estimates, we performed a two-dimensional umbrella sampling calculation with restraints on the $\phi$ and $\psi$ dihedral angles of the alanine dipeptide.  Parameters were the same as in the one-dimensional calculation above, with the addition of 20 bias functions in the $\psi$ dihedral with the same force constant, creating a grid of 400 windows.  Each window was equilibrated for 40 ps and sampled for a further 150 ps, with the collective variable values output every 10 fs.  

In Figures \ref{fig:2D_umb_fes} and \ref{fig:importance_C7ax_barrier}A, we plot the two-dimensional PMF from EMUS 
and the importances for the free energy difference between two windows located at the $C_7$ equatorial and $C_7$ axial configurations. Comparison shows that the importances are high for windows on low free energy pathways between the two windows of interest. Two such pathways exist.  In the representation in Figure \ref{fig:2D_umb_fes}, one proceeds up and to the left of the $C_7$ equatorial basin and then (via the periodic boundaries) enters the $C_7$ axial basin through transition state 1 (TS1 in Figure \ref{fig:2D_umb_fes}).  The other pathway proceeds down then right through transition state 2 (TS2 in Figure \ref{fig:2D_umb_fes}).  Of these two pathways, the first has a lower free energy barrier.  We observe that the EMUS importances are larger for windows located on this pathway.  In contrast, windows off these pathways in regions with high free energies are given very low importances.

\begin{figure*}
  \includegraphics[scale=1.0]{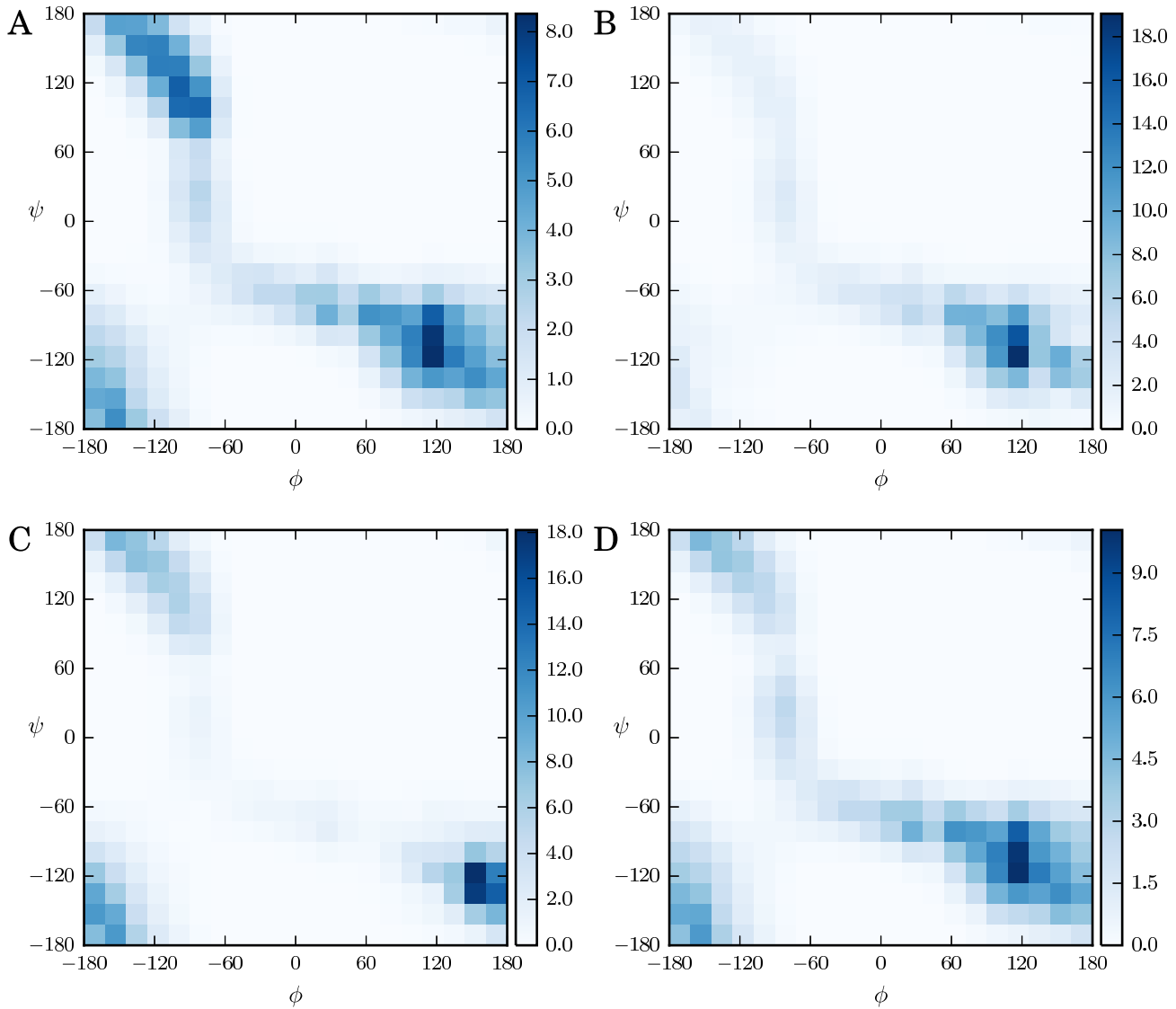}
  \caption{EMUS relative Importances. (A) Relative importances for the free energy difference between windows in the $C_7$ axial and $C_7$ equatorial basins.  The window in the $C_7$ equatorial basin is centered at (--81,81), and the window in the $C_7$ axial basin at (63,--63).  (B) Window importances for the free energy difference between windows in the $C_7$ axial basin and at TS1.  Windows are centered at (63,--63) and (135,$-117$), respectively.   (C) Importances for the free energy of the window at TS1.  (D) Importances for the window in the $C_7$ axial basin.}
  \label{fig:importance_C7ax_barrier}
\end{figure*}

We expect the importances to depend on the computed average.  To illustrate that this is the case numerically,
we show the log importances for the free energy difference between a window in the $C_7$ axial basin and one located on TS1 in Figure \ref{fig:importance_C7ax_barrier}B.
Compared to Figure \ref{fig:importance_C7ax_barrier}A, the importances are higher in the $C_7$ axial basin and lower in the  $C_7$ equatorial basin, which highlights that the importances depend on the average computed and do not simply mirror the free energy.
In Figures~\ref{fig:importance_C7ax_barrier}C and \ref{fig:importance_C7ax_barrier}D,
we plot the importances for estimating the window free energy (not the free energy difference)
of the window on TS1 and the window in the $C_7$ axial basin, respectively.
We note that the importances in the $C_7$ equatorial basin are higher in Figures \ref{fig:importance_C7ax_barrier}C and \ref{fig:importance_C7ax_barrier}D than in \ref{fig:importance_C7ax_barrier}B.
This suggests that when the free energy difference between the two windows is considered,
there is some cancellation of the errors arising in the $C_7$ equatorial basin.
%

\subsubsection{Comparison with Other Algorithms for Determining Error Contributions}\label{secn:zhuhummer}

Zhu and Hummer \cite{zhu2012convergence} proposed an algorithm for determining window free energies by calculating the mean restraining forces for each window and using thermodynamic integration to estimate free energy differences between adjacent windows. These are combined using least squares to calculate window free energies.  Like EMUS, this algorithm allows one to construct error estimates that can be decomposed into contributions from individual windows.  The authors give an expression for the error in the free energy of one window.  This expression can be easily extended to the free energy difference between two windows, giving
\begin{equation}\label{eq:zh_var_estimate}
\var \left(\Delta G_{ji}\right) = \sum_k^L \var \left(\sum_{\alpha}^{D}\left(c_{jk\alpha}-c_{ik\alpha}\right)\bar{f}_\alpha^k\right),
\end{equation}
where $\bar{f}_\alpha^k$ is the average force exerted by the bias function for window $k$ in the $\alpha$-th dimension.
The constants $c_{ik\alpha}$  and $c_{jk\alpha}$ are defined in Appendix~A of Zhu and Hummer~\cite{zhu2012convergence}.
The authors propose that these error estimates are applicable to WHAM and other umbrella sampling algorithms.

Both \eqref{eq:formula for asymptotic variance} and \eqref{eq:zh_var_estimate} are sums of contributions from individual windows.  Using the formalism introduced in Section~\ref{sec:estimate of asymptotic var}, we define the process
\begin{equation}\label{eq:defn_ZH_err_process}
  \zeta^{k,ZH}_t = \sum_{\alpha}^D \left(c_{jk\alpha}-c_{ik\alpha}\right) \left({f}_{\alpha}^k\left(X_t^k\right) -\<{f}_{\alpha,n}^k\>\right),
\end{equation}
and $\chi_i^{ZH}$ as the integrated autocovariance of $\zeta^{k,ZH}_t$.  This allows us to define importances for the Zhu and Hummer algorithm analogously to those for EMUS (see \eqref{eqn:relative importance}).

\begin{figure*}
  \includegraphics[scale=1.0]{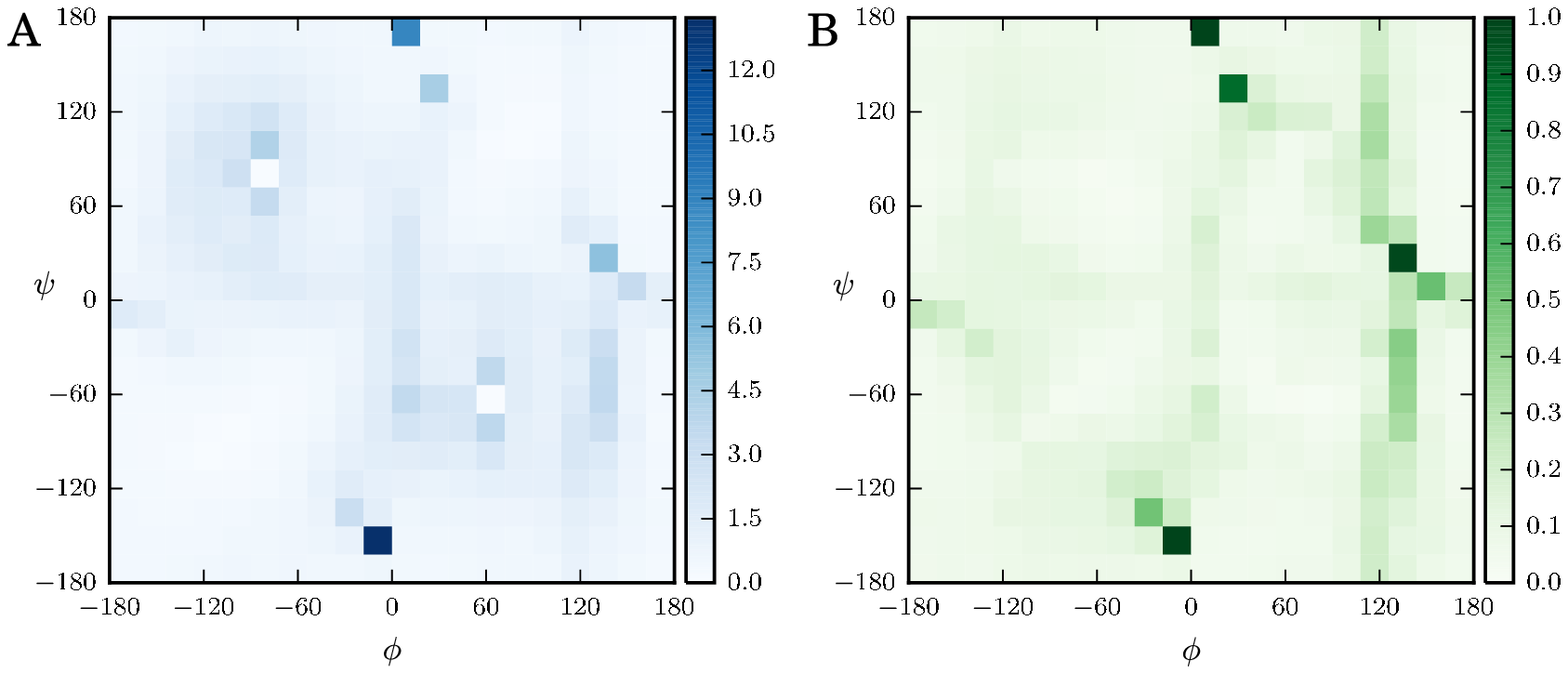}
  \caption{Comparison to Zhu and Hummer. (A) ZH estimates for the relative importances for the free energy difference between windows in the $C_7$ axial and $C_7$ equatorial basins.  Compare with Figure \ref{fig:importance_C7ax_barrier}A.  (B) Autocorrelation times of the trajectory $\zeta_t^{k,ZH}$ in each window.  The largest value observed is 3 ps, but the scale is limited to 1 ps for visual clarity.}
  \label{fig:znh_importance}
\end{figure*}

We applied the ZH error analysis to the two-dimensional umbrella sampling data used in Section \ref{sec:estimate of asymptotic var} and calculated the importances for the same free energy difference as in Figure~\ref{fig:2D_umb_fes} (Figure \ref{fig:znh_importance}A). Rather than falling along the low free energy pathways, as for EMUS, the ZH importances mirror the autocorrelation times (Figure \ref{fig:znh_importance}B).  This indicates that windows have large ZH importances if they have large fluctuations in free energy.  We thus see that different algorithms emphasize different windows in US.  We can understand the behaviors of these two algorithms by considering \eqref{eq:defn_EMUS_err_process} and \eqref{eq:defn_ZH_err_process}.
The factor $\partial B/\partial \avevec_i$ in \eqref{eq:defn_EMUS_err_process} depends explicitly on the normalization constant for each window (see Lemma~\ref{lem:formulas for partials of emus}).   By contrast, the factor $\left(c_{jk\alpha}-c_{ik\alpha}\right)$  in \eqref{eq:defn_ZH_err_process} depends only on the relative positions of the windows and not on their free energies.

\section{Conclusions} \label{sec:conclusions}
The success of an umbrella sampling simulation depends on the choice of windows (i.e., how the system is biased) and the estimator used to determine the normalization constants of the windows from trajectory data.  Here, we show that the normalization constants can be obtained from an eigenvector of a stochastic matrix.  This eigenvector method for umbrella sampling (EMUS) can be viewed as the first step in an implementation of the MBAR estimator.  In our experience, this first step is nearly converged, and machine precision is reached in only a few iterations.  Moreover, each iteration yields a consistent estimate. Most importantly, error analysis is considerably easier for EMUS than MBAR because the elements of the stochastic matrix do not depend on the normalization constants.

Within this framework, we revisited a common scaling argument for justifying umbrella sampling and showed that once the number of windows becomes  sufficiently large, the scheme does not benefit from the addition of more windows (i.e., the variance is not further reduced for a fixed computational effort).  We show that an alternative scaling regime in which temperature decreases (or, equivalently, free-energy barrier heights increase) as the number of windows increases best demonstrates the potential benefits of the umbrella sampling strategy;  in that regime the efficiency improvement over direct simulation is exponential  in the (inverse) temperature.

Our main theoretical result is a central limit theorem for the statistical averages obtained from EMUS.  This result relies on the delta method, which  we use to characterize the propagation of the asymptotic error through the solution of a stochastic matrix eigenproblem. The central limit theorem provides an expression for the asymptotic variance of the averages of interest.  It is a sum of contributions from individual windows, and we use it to develop a prescription for estimating the relative importances of windows for averages from the trajectory data.  For free energy differences of states of the alanine dipeptide, we find numerically that the importances are largest for low-free energy pathways that connect the specific states of interest.   These results suggest that the importances could serve as the basis for adaptive schemes that focus computational effort on the windows of most importance.  Even more interesting would be to adjust the bias functions as the simulation progresses.   How best to do this remains an open area of investigation.


\appendix
\section{Consistency of Iterative EMUS}\label{apdx: consistency}
Here, we prove that for fixed finite $m$, 
$z^m$ is a consistent estimator of the vector of normalization constants $z$.
With the initial guess $z^0 = n$, the result, $z^1$, of the first iteration is the EMUS estimator.
We now show that $z^2$ is also consistent in the sense that if the trajectory averages defining $\bar F(w)$  converge  then $z^2$ converges to $z.$  Because the various sequences in question are sequences of random variables one must specify what is meant by convergence.  The argument below applies when convergence refers either to convergence in probability or convergence with probability one (almost sure convergence) as long as the notion of convergence is consistent throughout.  Consistency of $z^m$ follows by induction on $m$ using a similar argument.

For any positive vector $w$, we define
\[
u_k = \prod_{j\neq k} \frac{w_j}{n_j},
\]
and we write
\[
\bar F_{ij}(w) = \frac{1}{N_i}\sum_{t=0}^{N_i-1} h_{ij}(u,x),
\]
where
\[
h_{ij}(u,x) =  \frac{ u_i \psi_j(x)}
{\sum_{k} u_k \psi_k(x) }.
\]
We then observe that
\[
\partial_{u_k} h_{ij}(u,x) =
 \begin{cases}
\frac{1}{u_i}  h_{ij}(u,x) [1 -h_{ii}(u,x)], &\text{if } k=i\\
-\frac{1}{u_i} h_{ij}(u,x)h_{ik}(u,x), & \text{if }k\neq i.
\end{cases}.
\]
Because $h_{ij}(u,x) \leq u_i / u_j,$
\[
\lvert \partial_{u_k} h_{ij}(u,x) \rvert \leq \frac{1}{u_j}\max\left\{1, \frac{u_i}{u_k}\right\}.
\]
Therefore,
\begin{equation}\label{eq: uniform modulus of continuity}
\lvert h_{ij}(\tilde{u},x) - h_{ij}(u,x)\rvert \leq \gamma(u,\tilde{u}),
\end{equation}
for a continuous function $\gamma$ defined for positive vectors $u$ and $\tilde{u}$
 and such that $\gamma(u,u)=0$ for any $u$.
The function $\gamma(u,\tilde{u})$ must explode when the entries of $u$ or $w$ approach $0$.
Now define
\begin{equation*}
u_k = \prod_{j\neq k} \frac{z^1_j}{n_j} \text{ and } \tilde{u}_k = \prod_{j\neq k} \frac{z_j}{n_j},
\end{equation*}
where $z$ is the exact vector of normalization constants.
By~\eqref{eq: uniform modulus of continuity}, we have
\begin{equation}\label{eq: uniform bound on Fz}
\lvert \bar F_{ij}(z^1) - \bar F_{ij}(z)\rvert \leq \gamma(u, \tilde{u}).
\end{equation}
As the number of samples $N$ increases, $\bar F(z)$ converges to $F(z).$
Moreover, since $z^1$ is the EMUS estimate of $z$,
$z^1$ converges to $z.$
Therefore, $u$ converges to $\tilde{u}$,
and~\eqref{eq: uniform bound on Fz} implies that $\bar F_{ij}(z^1)$ converges to $F_{ij}(z)$.
Finally, since the function mapping an irreducible, stochastic matrix to its invariant vector
is continuous, it follows that $z^2$ converges to the invariant vector of $F(z)$, which is $z$.
This verifies the consistency of $z^2$.

\section*{Acknowledgments}

This research was supported by National Institutes of Health (NIH) Grant Number 5 R01 GM109455-02.  We wish to thank Jonathan Mattingly, Jeremy Tempkin, and Charlie Matthews for helpful discussions.


\begin{thebibliography}{10}%
\makeatletter
\providecommand \@ifxundefined [1]{%
 \ifx #1\undefined \expandafter \@firstoftwo
 \else \expandafter \@secondoftwo
\fi
}%
\providecommand \@ifnum [1]{%
 \ifnum #1\expandafter \@firstoftwo
 \else \expandafter \@secondoftwo
\fi
}%
\providecommand \enquote [1]{``#1''}%
\providecommand \bibnamefont  [1]{#1}%
\providecommand \bibfnamefont [1]{#1}%
\providecommand \citenamefont [1]{#1}%
\providecommand\href[0]{\@sanitize\@href}%
\providecommand\@href[1]{\endgroup\@@startlink{#1}\endgroup\@@href}%
\providecommand\@@href[1]{#1\@@endlink}%
\providecommand \@sanitize [0]{\begingroup\catcode`\&12\catcode`\#12\relax}%
\@ifxundefined \pdfoutput {\@firstoftwo}{%
 \@ifnum{\z@=\pdfoutput}{\@firstoftwo}{\@secondoftwo}%
}{%
 \providecommand\@@startlink[1]{\leavevmode}%
 \providecommand\@@endlink[0]{}%
}{%
 \providecommand\@@startlink[1]{%
  \leavevmode
  \pdfstartlink
   attr{/Border[0 0 1 ]/H/I/C[0 1 1]}%
   user{/Subtype/Link/A<</Type/Action/S/URI/URI(#1)>>}%
  \relax
 }%
 \providecommand\@@endlink[0]{\pdfendlink}%
}%
\providecommand \url  [0]{\begingroup\@sanitize \@url }%
\providecommand \@url [1]{\endgroup\@href {#1}{\urlprefix}}%
\providecommand \urlprefix [0]{URL }%
\providecommand \Eprint[0]{\href }%
\@ifxundefined \urlstyle {%
  \providecommand \doi [1]{doi:\discretionary{}{}{}#1}%
}{%
  \providecommand \doi [0]{doi:\discretionary{}{}{}\begingroup
  \urlstyle{rm}\Url }%
}%
\providecommand \doibase [0]{http://dx.doi.org/}%
\providecommand \Doi[1]{\href{\doibase#1}}%
\providecommand \selectlanguage [0]{\@gobble}%
\providecommand \bibinfo [0]{\@secondoftwo}%
\providecommand \bibfield [0]{\@secondoftwo}%
\providecommand \translation [1]{[#1]}%
\providecommand \BibitemOpen[0]{}%
\providecommand \bibitemStop [0]{}%
\providecommand \bibitemNoStop [0]{.\EOS\space}%
\providecommand \EOS [0]{\spacefactor3000\relax}%
\providecommand \BibitemShut [1]{\csname bibitem#1\endcsname}%
\bibitem{torrie1977nonphysical}%
  \BibitemOpen
  \bibfield{author}{%
  \bibinfo {author} {\bibfnamefont{G.~M.}\ \bibnamefont{Torrie}}\ and\ \bibinfo
  {author} {\bibfnamefont{J.~P.}\ \bibnamefont{Valleau}},\ }%
  \bibfield{journal}{%
  \bibinfo {journal} {Journal of Computational Physics}\ }%
  \textbf{\bibinfo {volume} {23}},\ \bibinfo {pages} {187} (\bibinfo {year}
  {1977})\BibitemShut{NoStop}%
\bibitem{vardi1985empirical}%
  \BibitemOpen
  \bibfield{author}{%
  \bibinfo {author} {\bibfnamefont{Y.}~\bibnamefont{Vardi}},\ }%
  \bibfield{journal}{%
  \bibinfo {journal} {The Annals of Statistics},\ \bibinfo {pages} {178}}%
   (\bibinfo {year} {1985})\BibitemShut{NoStop}%
\bibitem{gill1988}%
  \BibitemOpen
  \bibfield{author}{%
  \bibinfo {author} {\bibfnamefont{R.~D.}\ \bibnamefont{Gill}}, \bibinfo
  {author} {\bibfnamefont{Y.}~\bibnamefont{Vardi}},\ and\ \bibinfo {author}
  {\bibfnamefont{J.~A.}\ \bibnamefont{Wellner}},\ }%
  \bibfield{journal}{%
  \Doi{10.1214/aos/1176350948}{\bibinfo {journal} {Ann. Statist.}}\ }%
  \textbf{\bibinfo {volume} {16}},\ \bibinfo {pages} {1069} (\bibinfo {month}
  {09}\ \bibinfo {year} {1988}),\
  \url{http://dx.doi.org/10.1214/aos/1176350948}\BibitemShut{NoStop}%
\bibitem{kumar1992weighted}%
  \BibitemOpen
  \bibfield{author}{%
  \bibinfo {author} {\bibfnamefont{S.}~\bibnamefont{Kumar}}, \bibinfo {author}
  {\bibfnamefont{D.}~\bibnamefont{Bouzida}}, \bibinfo {author}
  {\bibfnamefont{R.~H.}\ \bibnamefont{Swendsen}}, \bibinfo {author}
  {\bibfnamefont{P.~A.}\ \bibnamefont{Kollman}},\ and\ \bibinfo {author}
  {\bibfnamefont{J.~M.}\ \bibnamefont{Rosenberg}},\ }%
  \bibfield{journal}{%
  \bibinfo {journal} {Journal of Computational Chemistry}\ }%
  \textbf{\bibinfo {volume} {13}},\ \bibinfo {pages} {1011} (\bibinfo {year}
  {1992})\BibitemShut{NoStop}%
\bibitem{shirts2008statistically}%
  \BibitemOpen
  \bibfield{author}{%
  \bibinfo {author} {\bibfnamefont{M.~R.}\ \bibnamefont{Shirts}}\ and\ \bibinfo
  {author} {\bibfnamefont{J.~D.}\ \bibnamefont{Chodera}},\ }%
  \bibfield{journal}{%
  \bibinfo {journal} {The Journal of Chemical Physics}\ }%
  \textbf{\bibinfo {volume} {129}},\ \bibinfo {pages} {124105} (\bibinfo {year}
  {2008})\BibitemShut{NoStop}%
\bibitem{rosta2014free}%
  \BibitemOpen
  \bibfield{author}{%
  \bibinfo {author} {\bibfnamefont{E.}~\bibnamefont{Rosta}}\ and\ \bibinfo
  {author} {\bibfnamefont{G.}~\bibnamefont{Hummer}},\ }%
  \bibfield{journal}{%
  \bibinfo {journal} {Journal of Chemical Theory and Computation}\ }%
  \textbf{\bibinfo {volume} {11}},\ \bibinfo {pages} {276} (\bibinfo {year}
  {2014})\BibitemShut{NoStop}%
\bibitem{mey2014xtram}%
  \BibitemOpen
  \bibfield{author}{%
  \bibinfo {author} {\bibfnamefont{A.~S.}\ \bibnamefont{Mey}}, \bibinfo
  {author} {\bibfnamefont{H.}~\bibnamefont{Wu}},\ and\ \bibinfo {author}
  {\bibfnamefont{F.}~\bibnamefont{No{\'e}}},\ }%
  \bibfield{journal}{%
  \bibinfo {journal} {Physical Review X}\ }%
  \textbf{\bibinfo {volume} {4}},\ \bibinfo {pages} {041018} (\bibinfo {year}
  {2014})\BibitemShut{NoStop}%
\bibitem{tan2012theory}%
  \BibitemOpen
  \bibfield{author}{%
  \bibinfo {author} {\bibfnamefont{Z.}~\bibnamefont{Tan}}, \bibinfo {author}
  {\bibfnamefont{E.}~\bibnamefont{Gallicchio}}, \bibinfo {author}
  {\bibfnamefont{M.}~\bibnamefont{Lapelosa}},\ and\ \bibinfo {author}
  {\bibfnamefont{R.~M.}\ \bibnamefont{Levy}},\ }%
  \bibfield{journal}{%
  \bibinfo {journal} {The Journal of Chemical Physics}\ }%
  \textbf{\bibinfo {volume} {136}},\ \bibinfo {pages} {144102} (\bibinfo {year}
  {2012})\BibitemShut{NoStop}%
\bibitem{lelievre2010free}%
  \BibitemOpen
  \bibfield{author}{%
  \bibinfo {author} {\bibfnamefont{T.}~\bibnamefont{Leli{\`e}vre}}, \bibinfo
  {author} {\bibfnamefont{G.}~\bibnamefont{Stoltz}},\ and\ \bibinfo {author}
  {\bibfnamefont{M.}~\bibnamefont{Rousset}},\ }%
  \emph{\bibinfo {title} {Free {E}nergy {C}omputations: {A} {M}athematical
  {P}erspective}}\ (\bibinfo {publisher} {World Scientific},\ \bibinfo {year}
  {2010})\BibitemShut{NoStop}%
\bibitem{minh2009optimal}%
  \BibitemOpen
  \bibfield{author}{%
  \bibinfo {author} {\bibfnamefont{D.~D.}\ \bibnamefont{Minh}}\ and\ \bibinfo
  {author} {\bibfnamefont{J.~D.}\ \bibnamefont{Chodera}},\ }%
  \bibfield{journal}{%
  \bibinfo {journal} {The Journal of Chemical Physics}\ }%
  \textbf{\bibinfo {volume} {131}},\ \bibinfo {pages} {134110} (\bibinfo {year}
  {2009})\BibitemShut{NoStop}%
\bibitem{zhu2012convergence}%
  \BibitemOpen
  \bibfield{author}{%
  \bibinfo {author} {\bibfnamefont{F.}~\bibnamefont{Zhu}}\ and\ \bibinfo
  {author} {\bibfnamefont{G.}~\bibnamefont{Hummer}},\ }%
  \bibfield{journal}{%
  \bibinfo {journal} {Journal of Computational Chemistry}\ }%
  \textbf{\bibinfo {volume} {33}},\ \bibinfo {pages} {453} (\bibinfo {year}
  {2012})\BibitemShut{NoStop}%
\bibitem{golub1986using}%
  \BibitemOpen
  \bibfield{author}{%
  \bibinfo {author} {\bibfnamefont{G.~H.}\ \bibnamefont{Golub}}\ and\ \bibinfo
  {author} {\bibfnamefont{C.~D.}\ \bibnamefont{Meyer}, \bibfnamefont{Jr}},\ }%
  \bibfield{journal}{%
  \bibinfo {journal} {SIAM Journal on Algebraic Discrete Methods}\ }%
  \textbf{\bibinfo {volume} {7}},\ \bibinfo {pages} {273} (\bibinfo {year}
  {1986})\BibitemShut{NoStop}%
\bibitem{schneider1977concepts}%
  \BibitemOpen
  \bibfield{author}{%
  \bibinfo {author} {\bibfnamefont{H.}~\bibnamefont{Schneider}},\ }%
  \bibfield{journal}{%
  \bibinfo {journal} {Linear Algebra and its applications}\ }%
  \textbf{\bibinfo {volume} {18}},\ \bibinfo {pages} {139} (\bibinfo {year}
  {1977})\BibitemShut{NoStop}%
\bibitem{tan2004likelihood}%
  \BibitemOpen
  \bibfield{author}{%
  \bibinfo {author} {\bibfnamefont{Z.}~\bibnamefont{Tan}},\ }%
  \bibfield{journal}{%
  \bibinfo {journal} {Journal of the American Statistical Association}\ }%
  \textbf{\bibinfo {volume} {99}},\ \bibinfo {pages} {1027} (\bibinfo {year}
  {2004})\BibitemShut{NoStop}%
\bibitem{frenkel2001understanding}%
  \BibitemOpen
  \bibfield{author}{%
  \bibinfo {author} {\bibfnamefont{D.}~\bibnamefont{Frenkel}}\ and\ \bibinfo
  {author} {\bibfnamefont{B.}~\bibnamefont{Smit}},\ }%
  \emph{\bibinfo {title} {Understanding {M}olecular {S}imulation: {F}rom
  {A}lgorithms to {A}pplications}},\ Vol.~\bibinfo {volume} {1}\ (\bibinfo
  {publisher} {Academic press},\ \bibinfo {year} {2001})\BibitemShut{NoStop}%
\bibitem{geyer1992practical}%
  \BibitemOpen
  \bibfield{author}{%
  \bibinfo {author} {\bibfnamefont{C.~J.}\ \bibnamefont{Geyer}},\ }%
  \bibfield{journal}{%
  \bibinfo {journal} {Statistical Science},\ \bibinfo {pages} {473}}%
   (\bibinfo {year} {1992})\BibitemShut{NoStop}%
\bibitem{paliwal2013using}%
  \BibitemOpen
  \bibfield{author}{%
  \bibinfo {author} {\bibfnamefont{H.}~\bibnamefont{Paliwal}}\ and\ \bibinfo
  {author} {\bibfnamefont{M.~R.}\ \bibnamefont{Shirts}},\ }%
  \bibfield{journal}{%
  \bibinfo {journal} {Journal of Chemical Theory and Computation}\ }%
  \textbf{\bibinfo {volume} {9}},\ \bibinfo {pages} {4700} (\bibinfo {year}
  {2013})\BibitemShut{NoStop}%
\bibitem{shirts2007alchreview}%
  \BibitemOpen
  \bibfield{author}{%
  \bibinfo {author} {\bibfnamefont{M.~R.}\ \bibnamefont{Shirts}}, \bibinfo
  {author} {\bibfnamefont{D.~L.}\ \bibnamefont{Mobley}},\ and\ \bibinfo
  {author} {\bibfnamefont{J.~D.}\ \bibnamefont{Chodera}},\ }%
  \bibfield{journal}{%
  \bibinfo {journal} {Ann Rep Comput Chem}\ }%
  \textbf{\bibinfo {volume} {3}},\ \bibinfo {pages} {41} (\bibinfo {year}
  {2007})\BibitemShut{NoStop}%
\bibitem{abraham2015gromacs}%
  \BibitemOpen
  \bibfield{author}{%
  \bibinfo {author} {\bibfnamefont{M.~J.}\ \bibnamefont{Abraham}}, \bibinfo
  {author} {\bibfnamefont{T.}~\bibnamefont{Murtola}}, \bibinfo {author}
  {\bibfnamefont{R.}~\bibnamefont{Schulz}}, \bibinfo {author}
  {\bibfnamefont{S.}~\bibnamefont{P{\'a}ll}}, \bibinfo {author}
  {\bibfnamefont{J.~C.}\ \bibnamefont{Smith}}, \bibinfo {author}
  {\bibfnamefont{B.}~\bibnamefont{Hess}},\ and\ \bibinfo {author}
  {\bibfnamefont{E.}~\bibnamefont{Lindahl}},\ }%
  \bibfield{journal}{%
  \bibinfo {journal} {SoftwareX}\ }%
  \textbf{\bibinfo {volume} {1}},\ \bibinfo {pages} {19} (\bibinfo {year}
  {2015})\BibitemShut{NoStop}%
\bibitem{tribello2014plumed}%
  \BibitemOpen
  \bibfield{author}{%
  \bibinfo {author} {\bibfnamefont{G.~A.}\ \bibnamefont{Tribello}}, \bibinfo
  {author} {\bibfnamefont{M.}~\bibnamefont{Bonomi}}, \bibinfo {author}
  {\bibfnamefont{D.}~\bibnamefont{Branduardi}}, \bibinfo {author}
  {\bibfnamefont{C.}~\bibnamefont{Camilloni}},\ and\ \bibinfo {author}
  {\bibfnamefont{G.}~\bibnamefont{Bussi}},\ }%
  \bibfield{journal}{%
  \bibinfo {journal} {Computer Physics Communications}\ }%
  \textbf{\bibinfo {volume} {185}},\ \bibinfo {pages} {604} (\bibinfo {year}
  {2014})\BibitemShut{NoStop}%
\bibitem{mackerell2000development}%
  \BibitemOpen
  \bibfield{author}{%
  \bibinfo {author} {\bibfnamefont{A.~D.}\ \bibnamefont{MacKerell}}, \bibinfo
  {author} {\bibfnamefont{N.}~\bibnamefont{Banavali}},\ and\ \bibinfo {author}
  {\bibfnamefont{N.}~\bibnamefont{Foloppe}},\ }%
  \bibfield{journal}{%
  \bibinfo {journal} {Biopolymers}\ }%
  \textbf{\bibinfo {volume} {56}},\ \bibinfo {pages} {257} (\bibinfo {year}
  {2000})\BibitemShut{NoStop}%
\bibitem{ryckaert1977numerical}%
  \BibitemOpen
  \bibfield{author}{%
  \bibinfo {author} {\bibfnamefont{J.-P.}\ \bibnamefont{Ryckaert}}, \bibinfo
  {author} {\bibfnamefont{G.}~\bibnamefont{Ciccotti}},\ and\ \bibinfo {author}
  {\bibfnamefont{H.~J.}\ \bibnamefont{Berendsen}},\ }%
  \bibfield{journal}{%
  \bibinfo {journal} {Journal of Computational Physics}\ }%
  \textbf{\bibinfo {volume} {23}},\ \bibinfo {pages} {327} (\bibinfo {year}
  {1977})\BibitemShut{NoStop}%
\bibitem{grossfieldWHAM}%
  \BibitemOpen
  \bibfield{author}{%
  \bibinfo {author} {\bibfnamefont{A.}~\bibnamefont{Grossfield}},\ }%
  \enquote{\bibinfo {title} {{WHAM}: the weighted histogram analysis method
  (version 2.0.9)},}\  (\bibinfo {year} {2013}),\
  \url{http://membrane.urmc.rochester.edu/content/wham}\BibitemShut{NoStop}%
\bibitem{suli2003introduction}%
  \BibitemOpen
  \bibfield{author}{%
  \bibinfo {author} {\bibfnamefont{E.}~\bibnamefont{S{\"u}li}}\ and\ \bibinfo
  {author} {\bibfnamefont{D.~F.}\ \bibnamefont{Mayers}},\ }%
  \emph{\bibinfo {title} {An introduction to numerical analysis}}\ (\bibinfo
  {publisher} {Cambridge university press},\ \bibinfo {year}
  {2003})\BibitemShut{NoStop}%
\bibitem{chandler1987introduction}%
  \BibitemOpen
  \bibfield{author}{%
  \bibinfo {author} {\bibfnamefont{D.}~\bibnamefont{Chandler}},\ }%
  \bibfield{journal}{%
  \bibinfo {journal} {Introduction to Modern Statistical Mechanics, by David
  Chandler, pp. 288. Foreword by David Chandler. Oxford University Press, Sep
  1987. ISBN-10: 0195042778. ISBN-13: 9780195042771}\ }%
  \textbf{\bibinfo {volume} {1}} (\bibinfo {year} {1987})\BibitemShut{NoStop}%
\bibitem{chipot2007free}%
  \BibitemOpen
  \bibfield{author}{%
  \bibinfo {author} {\bibfnamefont{C.}~\bibnamefont{Chipot}}\ and\ \bibinfo
  {author} {\bibfnamefont{A.}~\bibnamefont{Pohorille}},\ }%
  \emph{\bibinfo {title} {Free {E}nergy {C}alculations}}\ (\bibinfo {publisher}
  {Springer},\ \bibinfo {year} {2007})\BibitemShut{NoStop}%
\bibitem{van1992computer}%
  \BibitemOpen
  \bibfield{author}{%
  \bibinfo {author} {\bibfnamefont{J.}~\bibnamefont{{van Duijneveldt}}}\ and\
  \bibinfo {author} {\bibfnamefont{D.}~\bibnamefont{Frenkel}},\ }%
  \bibfield{journal}{%
  \bibinfo {journal} {The Journal of Chemical Physics}\ }%
  \textbf{\bibinfo {volume} {96}},\ \bibinfo {pages} {4655} (\bibinfo {year}
  {1992})\BibitemShut{NoStop}%
\bibitem{comer2014adaptive}%
  \BibitemOpen
  \bibfield{author}{%
  \bibinfo {author} {\bibfnamefont{J.}~\bibnamefont{Comer}}, \bibinfo {author}
  {\bibfnamefont{J.~C.}\ \bibnamefont{Gumbart}}, \bibinfo {author}
  {\bibfnamefont{J.}~\bibnamefont{Hénin}}, \bibinfo {author}
  {\bibfnamefont{T.}~\bibnamefont{Lelièvre}}, \bibinfo {author}
  {\bibfnamefont{A.}~\bibnamefont{Pohorille}},\ and\ \bibinfo {author}
  {\bibfnamefont{C.}~\bibnamefont{Chipot}},\ }%
  \bibfield{journal}{%
  \bibinfo {journal} {The Journal of Physical Chemistry B}\ }%
  \textbf{\bibinfo {volume} {119}},\ \bibinfo {pages} {1129} (\bibinfo {year}
  {2014})\BibitemShut{NoStop}%
\bibitem{virnau2004calculation}%
  \BibitemOpen
  \bibfield{author}{%
  \bibinfo {author} {\bibfnamefont{P.}~\bibnamefont{Virnau}}\ and\ \bibinfo
  {author} {\bibfnamefont{M.}~\bibnamefont{M{\"u}ller}},\ }%
  \bibfield{journal}{%
  \bibinfo {journal} {The Journal of Chemical Physics}\ }%
  \textbf{\bibinfo {volume} {120}},\ \bibinfo {pages} {10925} (\bibinfo {year}
  {2004})\BibitemShut{NoStop}%
\bibitem{kelly2011reversibility}%
  \BibitemOpen
  \bibfield{author}{%
  \bibinfo {author} {\bibfnamefont{F.~P.}\ \bibnamefont{Kelly}},\ }%
  \emph{\bibinfo {title} {Reversibility and stochastic networks}}\ (\bibinfo
  {publisher} {Cambridge University Press},\ \bibinfo {year}
  {2011})\BibitemShut{NoStop}%
\bibitem{VanKotenArxiv}%
  \BibitemOpen
  \bibfield{author}{%
  \bibinfo {author} {\bibfnamefont{B.}~\bibnamefont{{Van Koten}}}\ and\
  \bibinfo {author} {\bibfnamefont{J.}~\bibnamefont{Weare}},\ }%
   \bibinfo {year} {in
  preparation}\BibitemShut{NoStop}%
\bibitem{LeLievre:FreeEnergyComp}%
  \BibitemOpen
  \bibfield{author}{%
  \bibinfo {author} {\bibfnamefont{T.}~\bibnamefont{Leli\`{e}vre}}, \bibinfo
  {author} {\bibfnamefont{G.}~\bibnamefont{Stoltz}},\ and\ \bibinfo {author}
  {\bibfnamefont{M.}~\bibnamefont{Rousset}},\ }%
  \emph{\bibinfo {title} {Free energy computations a mathematical
  perspective}}\ (\bibinfo {publisher} {Imperial College Press},\ \bibinfo
  {address} {Hackensack, N.J.},\ \bibinfo {year} {2010})\ p.\ \bibinfo {pages}
  {458}\BibitemShut{NoStop}%
\bibitem{bilodeau2008theory}%
  \BibitemOpen
  \bibfield{author}{%
  \bibinfo {author} {\bibfnamefont{M.}~\bibnamefont{Bilodeau}}\ and\ \bibinfo
  {author} {\bibfnamefont{D.}~\bibnamefont{Brenner}},\ }%
  \emph{\bibinfo {title} {Theory of {M}ultivariate {S}tatistics}}\ (\bibinfo
  {publisher} {Springer Science \& Business Media},\ \bibinfo {year}
  {2008})\BibitemShut{NoStop}%
\bibitem{perturbationArticle}%
  \BibitemOpen
  \bibfield{author}{%
  \bibinfo {author} {\bibfnamefont{E.}~\bibnamefont{Thiede}}, \bibinfo {author}
  {\bibfnamefont{B.}~\bibnamefont{{Van Koten}}},\ and\ \bibinfo {author}
  {\bibfnamefont{J.}~\bibnamefont{Weare}},\ }%
  \bibfield{journal}{%
  \bibinfo {journal} {SIAM Journal on Matrix Analysis and Applications}\ }%
  \textbf{\bibinfo {volume} {36}},\ \bibinfo {pages} {917} (\bibinfo {year}
  {2015})\BibitemShut{NoStop}%
\bibitem{vanderVaart:AsymptStat}%
  \BibitemOpen
  \bibfield{author}{%
  \bibinfo {author} {\bibfnamefont{A.~W.}\ \bibnamefont{{van der Vaart}}},\ }%
  \emph{\bibinfo {title} {Asymptotic Statistics}},\ Cambridge series on
  statistical and probabilistic mathematics.\ (\bibinfo {publisher} {Cambridge
  University Press},\ \bibinfo {address} {Cambridge ; New York},\ \bibinfo
  {year} {1998})\ p.\ \bibinfo {pages} {443}\BibitemShut{NoStop}%
\bibitem{acor}%
  \BibitemOpen
  \bibfield{author}{%
  \bibinfo {author} {\bibfnamefont{D.}~\bibnamefont{Foreman-Mackey}}\ and\
  \bibinfo {author} {\bibfnamefont{J.}~\bibnamefont{Goodman}},\ }%
  \enquote{\bibinfo {title} {{ACOR 1.1.1}},}\ \bibinfo {howpublished}
  {\url{https://pypi.python.org/pypi/acor/1.1.1}}\BibitemShut{NoStop}%
\end{thebibliography}

\providecommand{\noopsort}[1]{}\providecommand{\singleletter}[1]{#1}%

\end{document}